\documentclass[a4paper,12pt]{article}
\usepackage{test,subcaption} % Alexis' custom style file
\usepackage[authoryear,round,longnamesfirst]{natbib}

\usepackage[inline,shortlabels]{enumitem}
\setlist[enumerate,1]{label=(\roman*)}

\usepackage{color}
\usepackage[colorlinks,citecolor=blue,urlcolor=magenta]{hyperref}
\usepackage{doi}
\usepackage{xurl} % breaks long urls
\usepackage{booktabs} % for nicer tables
\usepackage{caption}
\usepackage{multirow}
\usepackage{graphicx}
\graphicspath{{../draft/figures/}} % set path to sibling directory

\usepackage[margin = 1.25in]{geometry}
\usepackage{setspace}
\allowdisplaybreaks

\usepackage{bbm}
\renewcommand{\id}{\mathbbm1} % indicator function
\renewcommand{\hat}{\widehat} % \widehat is nicer
\renewcommand{\Pr}{\operatorname{P}}

\title{Tuning Parameter-Free Nonparametric Density Estimation from Tabulated Summary Data}

\author{\normalsize
Ji Hyung Lee\thanks{\footnotesize\setlength{\baselineskip}{4.4mm} 
Ji Hyung Lee: \href{mailto:jihyung@illinois.edu}{jihyung@illinois.edu}. Department of Economics, University of Illinois, 214 David Kinley Hall, 1407 West Gregory Drive, Urbana, IL 61801, USA}
\and\normalsize
Yuya Sasaki\thanks{\footnotesize\setlength{\baselineskip}{4.4mm} Yuya Sasaki: \href{mailto:yuya.sasaki@vanderbilt.edu}{yuya.sasaki@vanderbilt.edu}. Department of Economics, Vanderbilt University, VU Station B \#351819, 2301 Vanderbilt Place, Nashville, TN 37235-1819, USA\smallskip}
\and\normalsize
Alexis Akira Toda\thanks{\footnotesize\setlength{\baselineskip}{4.4mm} Alexis Akira Toda: \href{mailto:atoda@ucsd.edu}{atoda@ucsd.edu}. Department of Economics, University of California San Diego, 9500 Gilman Dr, \#0508, La Jolla, CA 92093-0508, USA\smallskip}
\and\normalsize
Yulong Wang\thanks{\footnotesize\setlength{\baselineskip}{4.4mm} Yulong Wang: \href{mailto:ywang402@syr.edu}{ywang402@syr.edu}. Department of Economics, Syracuse University, 110 Eggers Hall, Syracuse, NY 13244-1020, USA}
}

\date{\today}

\numberwithin{equation}{section}
\begin{document}
\maketitle

\begin{abstract}
Administrative data are often easier to access as tabulated summaries than in the original format due to confidentiality concerns. Motivated by this practical feature, we propose a novel nonparametric density estimation method from tabulated summary data based on maximum entropy and prove its strong uniform consistency. Unlike existing kernel-based estimators, our estimator is free from tuning parameters and admits a closed-form density that is convenient for post-estimation analysis. We apply the proposed method to the tabulated summary data of the U.S. tax returns to estimate the income distribution.


\end{abstract}

\newpage

%%%%%%%%%%%%%%%%%%%%%%
\section{Introduction}\label{sec:Intro}
%%%%%%%%%%%%%%%%%%%%%%

Researchers can often access a tabulated summary of data more easily than the original data containing confidential information at the individual level. Examples include administrative tax data containing detailed information about individual-level records of income.
Tax authorities often release summary statistics of the income distributions in a tabulated format, such as the number of taxpayers and their average income, grouped by bins of income levels.
Despite their lack of details, such tabulated summary data are still useful for researchers in analyzing income distributions, especially for historically old times for which micro data are no longer available.

A typical econometric method for estimating the cross-sectional distribution of a continuous random variable, such as the kernel density estimator and the empirical cumulative distribution function, requires individual-level information, and thus is not suitable when researchers only have access to tabulated summaries.
In this paper, we propose a novel method for nonparametrically estimating the probability density function of an absolutely continuous distribution from tabulated summary data based on the maximum entropy (ME) principle.
This ME density estimator enjoys a number of desirable properties.
First, it is piecewise exponential, which allows for analytical post-estimation integration to calculate cumulative distribution estimates such as the mean, variance, top income shares, Lorenz curve, and the Gini coefficient.
Second, and more importantly, unlike alternative kernel-based estimators \citep{BlowerKelsall2002,Sun2014}, the ME density estimator is free from tuning parameters such as the bandwidth.
This feature is attractive in practice and provides a complete theoretical justification of the method, unlike the existing kernel-based methods for which the theory fails to formally account for the tuning parameter choice in practice.
We establish the strong uniform consistency for the ME density and cumulative distribution estimators under the asymptotic framework in which the resolution of bins in the table becomes finer at certain rates as the underlying sample size increases.

To illustrate our proposed method, we consider two applications. The first is a simulation study of the calculation of top income shares. Compared to the existing methods, our proposed method generates smaller bias and root mean squared error (RMSE). The second is the estimation of the income distribution and top income shares from the tabulated summaries of the U.S. tax returns data. Unlike the popular Pareto interpolation method of \cite{Piketty2003}, which relies on parametric assumptions in the upper tail, our method allows for the estimation in mid-sample as well as in the tails without parametric assumptions.

\paragraph{Related literature}
Our paper is related to a large literature in economics on inequality measures as well as in statistics and econometrics of estimation with grouped data.

Applied researchers working with income inequality measures have long been studying the interpolation problem from grouped data; see for instance \cite{CowellMehta1982} for an early review. The parametric methods for estimating the Lorenz curve (and hence computing the Gini coefficient) of \cite{KakwaniPodder1976} and \cite{VillasenorArnold1989} have been used by the World Bank. \cite{FeenbergPoterba1993} and \cite{Piketty2003} interpolate the upper tail of the income distribution by the Pareto distribution with local Pareto exponents estimated from a tabulated summary. See \cite{PikettySaez2003} for an application to the U.S. income distribution. More recently, \cite{BlanchetFournierPiketty2022} apply spline interpolation to what they call the inverted Pareto coefficients. Compared to this applied literature, our approach enjoys several advantages such as that
\begin{enumerate*}
\item the method is supported by the rigorous econometric theory that guarantees the performance in large samples,
\item the method is nonparametric across the entire distribution of income, and
\item the output is a closed-form density that is analytically tractable and convenient for post-estimation analysis.
\end{enumerate*}

In terms of estimating the income distribution with tabulated data, our proposed method differs from the existing methods, which can be categorized into parametric and nonparametric ones. 
For the former, \citet{Hajargasht2012} assume a parametric income distribution and develop a generalized method of moment (GMM) estimator for the unknown parameters. 
See also \citet{Chen2018} and \citet{HajargashtGriffiths2020} for other GMM methods relying on parametric assumptions.  
Despite their good fit in some data sets \citep[e.g.,][]{Jorda2021}, parametric methods in general suffer from model misspecifications. 
Specifically, although the estimator of the coefficients in the parametric model might still converge to some pseudo true values, the implied density estimator is in general inconsistent. 
The theoretical effect of misspecification on further estimations of other features, say moments and top income shares, are unknown. 
They may exhibit large biases in finite samples, as we show by simulation studies in Section \ref{sec:Simulation}.  
In contrast, our estimator is nonparametric and relatively robust to parametric assumptions. 

For nonparametric estimation with binned/grouped data, \citet{ScottSheather1985} study the kernel density estimator when the observations are equally spaced. 
On the other hand, we do not require data to be equally spaced.
Without the equal spacing restriction, \cite{BlowerKelsall2002} propose an alternative estimator with Gaussian kernel function, which \cite{Sun2014} extends by allowing for other kernel functions. 
These kernel approaches require a bandwidth as a tuning parameter, whose choice is challenging in practice especially under the nonstandard sampling setup of binned/grouped/tabulated data.
In contrast, our proposed ME estimator is free from tuning parameters and therefore more attractive in practice.
Furthermore, this tuning-parameter-free feature of our proposed method provides a complete theoretical justification under weak regularity condition. 
\cite{Reyes2016} derive the orders of magnitude of the bias and the variance of \citeauthor{ScottSheather1985}'s estimator under very strong conditions:
\begin{enumerate*}
\item the bin size is of a smaller order than $h^2$, where $h$ denotes the bandwidth of the kernel estimator and
\item the underlying distribution function is seven-times differentiable with bounded derivatives.
\end{enumerate*} 
The first condition implies that the group structure is negligible and hence the bias and the variance are asymptotically the same as in the case with individual observations. 
The second condition is strong and rules out some candidate distributions, such as double Pareto, which is used in empirical studies of income distribution.  
Unlike the existing kernel-based methods, our ME estimator only requires the underlying density to be Lipschitz continuous and is based on minimizing the Kullback-Leibler divergence, which lead to both numerical and theoretical advantages. See Section \ref{sec:main} for details. 

Finally, our paper is related to the large literature that applies the maximum entropy (ME) principle. Historically, the ME principle was developed in physics to infer the population distribution (\eg, energy distribution of gas molecules) from macroscopic variables (\eg, temperature); see \cite{jaynes1957a}.  In economics, applications of the ME method include general equilibrium theory \citep{foley1994,Toda2010ET,Toda2015ET}, diagnosis of asset pricing models \citep{stutzer1995}, derivative pricing \citep{stutzer1996}, discretization of probability distributions and stochastic processes \citep{TanakaToda2013EL,TanakaToda2015SINUM,FarmerToda2017QE}, among others. In this paper, we apply the ME principle as a tool to impose moment restrictions implied by the tabulated summary data. Applications in econometrics include \cite{kitamura-stutzer1997} and \cite{Wu2003}, among others. 
To our knowledge, none of the existing ME methods work for tabulated summary data.

\paragraph{Organization of the paper} 
Section \ref{sec:TD} introduces the general data framework and previews the U.S. tax return data set used in the application. Section \ref{sec:main} introduces the nonparametric density estimator and proves its strong uniform consistency. Section \ref{sec:Simulation} presents simulation studies. Section \ref{sec:Income} applies the proposed method to the U.S. tax return data set.

%%%%%%%%%%%%%%%%%%%%%%
\section{Tabulated summary data}\label{sec:TD}
%%%%%%%%%%%%%%%%%%%%%%
\subsection{General data framework}\label{sec:TD:general}
%%%%%%%%%%%%%%%%%%%%%%

Consider the latent sample $\set{Y_i}_{i=1}^n$, which is not directly observed by the researcher. 
Denote the order statistics in descending order by $\set{Y_{(i)}}_{i=1}^n$ such that
\begin{equation*}
Y_{(1)}\ge Y_{(2)}\ge \dots\ge Y_{(n)}.
\end{equation*}
For each $m \in \set{1,\dots, n}$, define the partial sum of top $m$ order statistics
\begin{equation}
S_m\coloneqq \sum_{i=1}^m Y_{(i)}. \label{eq:partial_sum}
\end{equation}
Consider a positive number $K$ of bins, where $\set{t_k}_{k=0}^K$ denotes the sequence of bin threshold values such that
\begin{equation*}
\infty=t_0>t_1>t_2>\dots>t_K.
\end{equation*}
Let $n_k$ be the number of order statistics included in the top $k$ bins, that is, $n_k\coloneqq \#\set{i:Y_{(i)} \ge t_k}$. The tabulated data are summarized as $\set{(t_k,n_k,S_{n_k})}_{k=1}^K$, which is observed by the researcher.

%%%%%%%%%%%%%%%%%%%%%%
\subsection{Example: summary of income data by tax authorities}\label{sec:TD:income}
%%%%%%%%%%%%%%%%%%%%%%

As a concrete example of the data framework just described, consider the latent sample $\set{Y_i}_{i=1}^n$ of the values $Y_i$ of income, where $i$ indexes potential taxpayers and $n$ is the sample size. Due to confidentiality concerns, in general there is no public access to administrative data of income. Publicly available data on the income distribution released from tax authorities often take the form of the tabulated summary $\set{(t_k,n_k,S_{n_k})}_{k=1}^K$.

Table \ref{t:IRS2019} presents an example data set from the 2019 U.S.\ tax returns.\footnote{\label{fn:IRS}Table \ref{t:IRS2019} shows partial information from Internal Revenue Service, Statistics of Income (SOI) Individual Income Tax Returns Publication 1304 (\url{https://www.irs.gov/statistics/soi-tax-stats-individual-income-tax-returns-complete-report-publication-1304}), Table 1.4 under ``Basic Tables''. Adjusted gross income (AGI) is AGI less deficit. We omit the row corresponding to negative income.} 
In this example, the number of income groups is $K=18$. 
Column (1) shows the lower threshold $t_k$ of adjusted gross income (AGI) for each income group $k \in \set{1,\dots,K}$. 
Column (2) shows the number of taxpayers within each income group $k$, which corresponds to $n_k-n_{k-1}$ in our notations, where we set $n_0=0$ by convention. 
Column (3) shows the total income (AGI) accruing to taxpayers in each income group $k$ in units of 1,000 U.S.\ dollars, which corresponds to $(S_{n_k}-S_{n_{k-1}})/1{,}000$ in our notations, where we set $S_0=0$ by convention. 
We thus observe the tabulated summary data $\set{(t_k,n_k,S_{n_k})}_{k=1}^K$ of AGI, where $K=18$.

\begin{table}[!htb]
\centering
\caption{Income distribution in the United States, 2019.}\label{t:IRS2019}
\begin{tabular}{rrrr}
\toprule
\multicolumn2{c}{Income group} & \multicolumn2{c}{Adjusted gross income (AGI)} \\
& \multicolumn1{c}{(1)} & \multicolumn1{c}{(2)} & \multicolumn1{c}{(3)} \\
$k$ & AGI threshold & \# returns & Total income \\
\cmidrule(lr){1-2}
\cmidrule(lr){3-4}
18 & \$1 &	9,866,880	&	24,439,988 \\
17 & \$5,000 &	9,925,940	&	74,584,857 \\
16 & \$10,000 &	11,087,737	&	138,230,399	\\
15 & \$15,000 &	10,039,446	&	175,255,963	\\
14 & \$20,000 &	9,493,968	&	213,660,160	\\
13 & \$25,000 &	9,289,939	&	254,877,708	\\
12 & \$30,000 &	16,090,602	&	560,073,192	\\
11 & \$40,000 &	12,503,041	&	560,258,808	\\
10 & \$50,000 &	22,238,948	&	1,366,892,948 \\
9 & \$75,000 &	14,118,568	&	1,222,947,425 \\
8 & \$100,000 &	21,997,582	&	3,004,363,636 \\
7 & \$200,000 &	7,297,883	&	2,090,808,696 \\
6 & \$500,000 &	1,162,371	&	781,920,814	\\
5 & \$1,000,000 &	254,197	&	305,561,848	\\
4 & \$1,500,000 &	103,075	&	176,961,208	\\
3 & \$2,000,000 &	143,514	&	425,088,995	\\
2 & \$5,000,000 &	34,738	&	237,781,553	\\
1 & \$10,000,000 &	20,876	&	590,230,011	\\
\bottomrule
\end{tabular}
\caption*{\footnotesize Note: ``AGI threshold'' is the lower threshold of adjusted gross income (AGI) that defines the income groups. ``\# returns'' is the number of tax returns with income weakly above the lower AGI threshold and strictly below the upper AGI threshold. ``Total income'' is the total income (in units of 1,000 U.S.\ dollars) accruing to taxpayers in each income group.}
\end{table}

%%%%%%%%%%%%%%%%%%%%%%
\section{Main results}\label{sec:main}
%%%%%%%%%%%%%%%%%%%%%%

This section investigates a method to characterize a well-behaved density function of the distribution of $Y$ from the tabulated summary data $\set{(t_k,n_k,S_{n_k})}_{k=1}^K$ introduced in Section \ref{sec:TD}. We first consider the case when the sample size is infinite and there are no sampling errors in bin probabilities and conditional means. We next propose a feasible estimator and study its asymptotic properties as the sample size tends to infinity.

\subsection{Maximum entropy density}\label{subsec:ME}

Let $F$ denote the true cumulative distribution function (CDF) of $Y$, which is assumed to be absolutely continuous with probability density function denoted by $f = F'$. 
Suppose that the thresholds satisfy
\begin{equation*}
\infty=t_0>t_1>t_2>\dots>t_K \eqqcolon \ubar{t},
\end{equation*}
and let $I_k\coloneqq [t_k,t_{k-1})$ denote the interval for the top $k$-th bin with top fractile denoted by $p_k=\Pr(Y\ge t_k)=1-F(t_k)$. 
For each $k$, the bin probability and conditional mean are defined by
\begin{subequations}\label{eq:qyk}
\begin{align}
q_k&\coloneqq \Pr(Y \in I_k)=\int \id_{I_k}(y)\diff F(y)=p_k-p_{k-1}, \label{eq:qk}\\
y_k&\coloneqq \E[Y \mid Y \in I_k]=\frac{1}{q_k}\int \id_{I_k}(y)y\diff F(y),\label{eq:yk}
\end{align}
\end{subequations}
where $\id_{I}(\cdot)$ denotes the indicator function indicating that the argument belongs to the set $I$.

Obviously, given only the finite tabulation $\set{(p_k,y_k)}_{k=1}^K$, we do not have sufficient moment restrictions to pin down the true density function $f$.  The maximum entropy (ME) method is useful when only certain moment conditions are given.  In our context of characterizing the distribution of $Y$ from a tabulation, we can proceed as follows.

Letting $g$ denote a generic density, the given moment conditions consistent with \eqref{eq:qyk} are
\begin{subequations}\label{eq:mcond}
\begin{align}
\int \id_{I_k}(y)g(y)\diff y&=q_k, \label{eq:mcond_prob}\\
\int \id_{I_k}(y)yg(y)\diff y&=q_ky_k \label{eq:mcond_mean}
\end{align}
\end{subequations}
for each $k \in \set{1,\dots,K}$. 
The ME density $f^*$ is defined by the density $g$ on $I\coloneqq [\ubar{t},\infty)$ that minimizes the Kullback-Leibler divergence (with respect to the improper uniform density)
\begin{equation}
\int_I g(y)\log g(y)\diff y \label{eq:KL}
\end{equation}
subject to the moment restrictions \eqref{eq:mcond}. 
Below, we let $L_+^1(I)$ denote the equivalence class (identified by the $L^1$ norm) of nonnegative, measurable, and integrable functions $g:I\to \R$.
The following proposition characterizes the solution to the ME problem.

\begin{prop}\label{prop:ME_density}
Let $\infty=t_0>t_1>\dots>t_K$ be the thresholds and $y_k\in (t_k,t_{k-1})$ be the average of $Y$ in group $k$. Then the function
\begin{equation}
J_k(\lambda;t)\coloneqq
\begin{cases}
y_k\lambda-\log\left(\frac{\e^{\lambda t_{k-1}}-\e^{\lambda t_k}}{\lambda}\right) & (\lambda\neq 0)\\
-\log \left(t_{k-1}-t_k\right) & (\lambda=0)
\end{cases} \label{eq:Jk}
\end{equation}
is strictly concave in $\lambda\in \R$ and achieves a unique maximum $\lambda_k^*$, which satisfies
\begin{equation}
\lambda_1^*=-\frac{1}{y_1-t_1}<0\quad \text{and}\quad 
\lambda_k^*\begin{cases*}
<0 & if $y_k<\frac{t_k+t_{k-1}}2$,\\
=0 & if $y_k=\frac{t_k+t_{k-1}}2$,\\
>0 & if $y_k>\frac{t_k+t_{k-1}}2$,
\end{cases*}
\quad \text{for $k\ge 2$}.\label{eq:lambdak_sign}
\end{equation}
The ME problem has a unique solution $f^*\in L_+^1(I)$, which is piecewise exponential and satisfies
\begin{equation}
f^*(y)=\begin{cases}
\frac{q_k\lambda_k^*\e^{\lambda_k^* y}}{\e^{\lambda_k^*t_{k-1}}-\e^{\lambda_k^*t_k}} & (\lambda_k^*\neq 0)\\
\frac{q_k}{t_{k-1}-t_k} & (\lambda_k^*=0)
\end{cases}\label{eq:fstar}
\end{equation}
for $y\in I_k=[t_k,t_{k-1})$. The minimum value of the Kullback-Leibler divergence \eqref{eq:KL} is given by
\begin{equation}
J^*(t_1,\dots,t_K)\coloneqq \sum_{k=1}^K q_k(J_k(\lambda_k^*;t)+\log q_k). \label{eq:Jstar}
\end{equation}
\end{prop}

\subsection{Estimation}\label{subsec:estim}

This section proposes a feasible analog of the ME density characterized in Section \ref{subsec:ME} for estimation of the true density function $f$ of $Y$. We then establish its strong uniform consistency.

We construct a feasible ME density estimator by replacing $q_k$ and $y_k$ with their empirical analogs $\hat{q}_k=(n_k-n_{k-1})/n$  and $\hat{y}_k=(S_{n_k}-S_{n_{k-1}})/(n_k-n_{k-1})$, respectively. 
Letting $t \coloneqq \set{t_k}^K_{k=1}$ denote the vector of thresholds, we thus define the sample-analog ME estimator $\hat{f}$ of $f$ as the solution to the constrained optimization prblem of minimizing \eqref{eq:KL} subject to \eqref{eq:mcond} with $\hat{Q}_k$ and $\hat y_k$ in place of  $q_k$ and $y_k$, respectively.

We now establish the almost sure uniform consistency of $\hat{f}$ for the true density function $f$ over any compact subset of the domain of $f$ as $n\to\infty$. To this end, consider the following conditions.

\begin{asmp}\label{asmp:primitive}
The following conditions hold.
\begin{enumerate}
\item \label{cond:iid} $\set{Y_i}_{i=1}^n$ is \iid with density $f$.
\item \label{cond:f} The density $f:\R \to [0,\infty)$ is Lipschitz continuous with constant $L>0$, that is, for all $y_1,y_2\in \R$ we have
\begin{equation*}
	\abs{f(y_1)-f(y_2)}\le L\abs{y_1-y_2}.
\end{equation*}
\item \label{cond:KN} There exist some constants $c_1,c_2>0$ and $0<r_2\le r_1<1/4$ such that
\begin{equation}
c_1n^{-r_1}\le \min_{2\le k\le K}\set{t_{k-1}-t_k}\le \max_{2\le k\le K}\set{t_{k-1}-t_k}\le c_2n^{-r_2}. \label{eq:tdiff_order}
\end{equation}
Furthermore, letting $D\coloneqq \set{y\in \R: f(y)>0}$  be the domain of $f$, we have $\limsup_{n\to\infty} t_K\le \inf D$ and $\liminf_{n\to\infty} t_1\ge \sup D$.
\end{enumerate}
\end{asmp}

Condition \ref{cond:iid} assumes a random sample. 
Condition \ref{cond:f} implies that $f$ is almost everywhere differentiable with a bounded derivative, which is much weaker than typical assumptions on kernel estimators that require high-order smoothness conditions. 
Condition \ref{cond:KN} requires that the length of any bin is neither too large nor too small, as well as that the interval $[\limsup_{n \to \infty}t_K, \liminf_{n \to \infty}t_1]$ covers the domain $D$ of $f$.  
With these conditions, the following theorem establishes the strong uniform consistency of the maximum entropy density estimator $\hat{f}$ for the true density function $f$ over any compact subset of the domain and also bounds the convergence rate. The proof is non-trivial and deferred to Appendix \ref{sec:proof}.

\begin{thm}\label{thm:consistent}
Suppose that Assumption \ref{asmp:primitive} holds and let $C\subset D$ be compact. Then, as $n\to \infty$, we have
\begin{subequations}
\begin{align}
\sup_{y \in C} \abs{\hat{f}(y) - f(y)} &\asto 0, \label{eq:consistent}\\
\sup_{y \in C} \abs{\hat{f}(y) - f(y)} &= \underbrace{O(n^{-r_2})}_\text{deterministic}+\underbrace{O_p\left(n^{-\frac{1-4r_1}{2}}\right)}_\text{stochastic}. \label{eq:rate}
\end{align}
\end{subequations}
\end{thm}

A few remarks are in order regarding this result on the convergence rate.
First, while the rate is decomposed into the the deterministic part and the stochastic part, we do not have a control over the trade-off between these two components in the absence of a tuning parameter.
This implies a drawback of the tuning parameter-free approach.
Second, the rate depends on the parameters $r_1$ and $r_2$ of bin lengths.
Slowly vanishing bin lengths (i.e., small $r_1$ and $r_2$) yield small variances at the expense of large biases.
Quickly vanishing bin lengths (i.e., large $r_1$ and $r_2$) yield small biases at the expense of large variances.
Third, suppose $r_1=r_2$ for simplicity.
Then, $1/6 < r_1=r_2 < 1/4$ implies that the stochastic part dominates, while $r_1=r_2 < 1/6$ implies that the deterministic part dominates.
In the latter case, the limit distribution has a biased center, and it is difficult to conduct statistical inference in general.
This is another limitation of the tuning parameter-free approach. 

\cite{Reyes2016} derive the convergence rate of the kernel estimator proposed by \cite{ScottSheather1985}. 
\cite{Reyes2016} assume that the bin lengths are $o(h_n^2)$, where $h_n$ denotes the bandwidth satisfying $h_n\to 0$ and $nh_n \to\infty$ as $n\to\infty$. 
This condition implies that the group/bin structure is asymptotically negligible, and the resulting orders, $O(h_n^2)$ and $O_p(1/\sqrt{nh_n})$, of the deterministic and stochastic parts, respectively, are the same as those in the standard case with individual observations. 
On the one hand, choosing a certain bandwidth $h_n$ could lead to a smaller bias or variance than our ME estimator. 
On the other hand, the assumption that the bin lengths are $o(h_n^2)$ is very restrictive and could be violated in empirical studies where the bins are not too small. 
Furthermore, \cite{Reyes2016} assume that the underlying distribution function is seven-times differentiable with bounded derivatives. 
In contrast, our ME estimator only requires Lipschitz continuity for $f$, which is another advantage. 

Having established the strong uniform consistency of the density, it is straightforward to establish the same for the cumulative distribution function (CDF) and quantiles. Define the estimator $\hat{F}$ of $F$ by
\begin{equation*}
\hat{F}(y) = \int_{-\infty}^y \hat{f}(y)\diff y.
\end{equation*}
The following corollary shows the strong uniform consistency of $\hat{F}$.

\begin{cor}\label{cor:cdf}
Suppose that Assumption \ref{asmp:primitive} holds and let $C\subset D$ be compact.  If $-\infty < \inf D$, then,
as $n\to\infty$, we have
\begin{equation*}
\sup_{y \in C} \abs{\hat{F}(y) - F(y)} \asto 0.
\end{equation*}
\end{cor}

Let $Q_\tau = \inf\set{ y : \tau \le F(y)}$ denote the $\tau$-th quantile of $F$.
Given $\hat{F}$, we can estimate $Q_\tau$ by the analog $\hat{Q}_\tau = \inf\set{ y : \tau \le \hat{F}(y)}$.
The following corollary shows that this quantile estimator is also consistent.

\begin{cor}\label{cor:quantile}
Suppose that Assumption \ref{asmp:primitive} holds and let $C\subset D$ be compact. If $-\infty < \inf D$ and $Q_\tau \in \interior C$, then, as $n\to \infty$, we have 
$\hat{Q}_\tau \asto Q_\tau$.
\end{cor}

\subsection{Discussion}\label{subsec:discussion}

In an early review of interpolation methods from grouped data of income, \cite{CowellMehta1982} list the following ten desirable properties (with slight rewording) that the hypothetical interpolated distribution should possess.
\begin{enumerate}[1]
\item\label{item:mcond} The bin probability and conditional mean ($q_k$ and $y_k$ in our notation in \eqref{eq:mcond}) of the estimated density $\hat{f}$ agree with the tabulated summary.
\item $\hat{f}\ge 0$.
\item $\hat{f}$ is continuous within any interval $I_k=[t_k,t_{k-1})$.
\item\label{item:cont} $\hat{f}$ is continuous.
\item\label{item:diff} $\hat{f}$ is differentiable.
\item $\lim_{y\to\infty}\hat{f}(y)=0$.
\item $\lim_{y\to\infty}\hat{f}'(y)=0$.
\item\label{item:turn} $\hat{f}$ has ``few'' turning points within each bin.
\item\label{item:range} The range of $\hat{f}$ is ``small'' on any given interval.
\item\label{item:closed} $\hat{f}$ admits a closed-form expression to compute inequality measures.
\end{enumerate}
In addition to these properties, we would like to add:
\begin{enumerate}[1,resume]
\item\label{item:conv} $\hat{f}$ converges to the true density $f$ as the sample size tends to infinity.
\end{enumerate}

The methods reviewed in \cite{CowellMehta1982} as well as those proposed thereafter satisfy only a few of these properties. For instance, the recent method of \cite{BlanchetFournierPiketty2022} does not satisfy \ref{item:turn} and \ref{item:range} (because it uses polynomial interpolation), \ref{item:closed} (because it interpolates the inverted Pareto coefficients, not the density), or \ref{item:conv} (they do not provide formal theorems).

In contrast, our ME density estimator $\hat{f}$ satisfies all properties except \ref{item:cont} and \ref{item:diff}. To see this, property \ref{item:mcond} holds by construction and \ref{item:conv} is established in Theorem \ref{thm:consistent}. All other properties (except \ref{item:cont} and \ref{item:diff}) hold because $\hat{f}$ is piecewise exponential explicitly given by \eqref{eq:fstar} and hence is nonnegative, continuously differentiable, and monotonic on each interval. Regarding properties \ref{item:cont} and \ref{item:diff}, they clearly hold except at bin thresholds.

To illustrate property \ref{item:closed} further, we present some integral formulas that are useful when computing the CDF and top income shares when the density is piecewise exponential. Consider the piecewise exponential density \eqref{eq:fstar}. To simplify the notation, let $q_k=q$, $\lambda_k^*=\lambda$, $t_k=a$, and $t_{k-1}=b$. Therefore, for $y\in [a,b)$, the density is
\begin{equation*}
f(y)=\begin{cases}
\frac{q\lambda\e^{\lambda y}}{\e^{\lambda b}-\e^{\lambda a}}, & (\lambda\neq 0)\\
\frac{q}{b-a}. & (\lambda=0)
\end{cases}
\end{equation*}
The counter CDF (tail probability) can be computed using
\begin{equation*}
\int_y^b f(y)\diff y=\begin{cases}
q\frac{\e^{\lambda b}-\e^{\lambda y}}{\e^{\lambda b}-\e^{\lambda a}}, & (\lambda\neq 0)\\
q\frac{b-y}{b-a}. & (\lambda=0)
\end{cases}
\end{equation*}
Applying integration by parts, the tail expectation can be computed using
\begin{align*}
\int_y^b yf(y)\diff y=\begin{cases}
\frac{q}{\e^{\lambda b}-\e^{\lambda a}}\left(b\e^{\lambda b}-y\e^{\lambda y}-\frac{1}{\lambda}(\e^{\lambda b}-\e^{\lambda y})\right),& (\lambda \neq 0)\\
\frac{q}{2}\frac{b^2-y^2}{b-a}. & (\lambda=0)
\end{cases}
\end{align*}
Putting all the pieces together, we obtain the following closed-form expressions for the CDF and tail expectation. (We assume $\lambda_k\neq 0$ for simplicity, and we use the notation $y\vee t=\max\set{y,t}$.)
\begin{align*}
&1-F(y)=\int_y^\infty f(y)\diff y=\sum_{k=1}^K\id(y< t_{k-1})q_k\frac{
\e^{\lambda_k t_{k-1}}-\e^{\lambda_k (y\vee t_k)}}{\e^{\lambda_k t_{k-1}}-\e^{\lambda_k t_k}},\\
&\int_y^\infty yf(y)\diff y
=\sum_{k=1}^K\id(y< t_{k-1})q_k\frac{(t_{k-1}-1/\lambda_k)\e^{\lambda_k t_{k-1}}-((y\vee t_k)-1/\lambda_k)\e^{\lambda_k(y\vee t_k)}}{\e^{\lambda_k t_{k-1}}-\e^{\lambda_k t_k}}.
\end{align*}

%%%%%%%%%%%%%%%%%%%%%%
\section{Simulation studies}\label{sec:Simulation}
%%%%%%%%%%%%%%%%%%%%%%

We conduct three simulation studies that examine the performance of our proposed ME method relative to existing methods.

\subsection{Density estimation}\label{subsec:sim_density}

We first present density estimates from one large simulation draw with the sample size comparable to those in our empirical data. 
We consider four typical models for income distribution, namely lognormal, gamma, Weibull, and double Pareto. In each case, we choose the scale parameter so that the population mean normalizes to 1. Table \ref{t:model} summarizes these distributions.

\begin{table}[!htb]
\centering${}$\\
\caption{Models of income distributions.}\label{t:model}
\begin{tabular}{lcll}
\toprule
Name & Density $f(y)$ & Normalization & Parameter(s) \\
\midrule
Lognormal & $\frac{1}{\sqrt{2\pi\sigma^2}}\e^{-\frac{(\log y-\mu)^2}{2\sigma^2}}\frac{1}{y}$ & $\mu=-\sigma^2/2$ & $\sigma=1.5$\\
Gamma & $\frac{b^a}{\Gamma(a)}y^{a-1}\e^{-by}$ &
$b=a$ & $a=0.5$\\
Weibull & $bk y^{k-1}\e^{-by^k}$ & $b=\Gamma(1+1/k)^k$ & $k=0.7$\\
Double Pareto & $\begin{cases}
\frac{\alpha\beta}{\alpha+\beta}M^{-\beta}y^{\beta-1} &(y\le M)\\
\frac{\alpha\beta}{\alpha+\beta}M^{\alpha}y^{-\alpha-1} &(y\ge M)
\end{cases}$ & $M=\frac{(\beta+1)(\alpha-1)}{\alpha\beta}$ & $(\alpha,\beta)=(1.5,0.5)$ \\
\bottomrule\\
\end{tabular}
\end{table}

The simulation design is as follows. For each model, we generate a random sample $\set{Y_i}_{i=1}^n$ with size $n=10^7$. 
Such a large $n$ is coherent with the number of tax payers in our empirical data set; see Table \ref{t:IRS2019}. 
We set the top fractiles to
\begin{equation*} (p_k)_{k=1}^K=(0.001,0.005,0.01,0.05,0.1,\dots,0.9,0.95,0.99,0.995,0.999,1)
\end{equation*}
(so $K=26$) and define the $k$-th threshold $t_k$ as the top $p_k$-th quantile of $\set{Y_i}_{i=1}^n$. We then estimate the ME density $\hat{f}$ as in Section \ref{sec:main}.

Figure \ref{fig:sim} shows the population and estimated densities for each model. 
Because the population densities are skewed, for visibility we plot the density of $\log Y$. 
In each case, the two densities $f$ and $\hat{f}$ are nearly identical.

\begin{figure}[!htb]
\centering
\includegraphics[width=\linewidth]{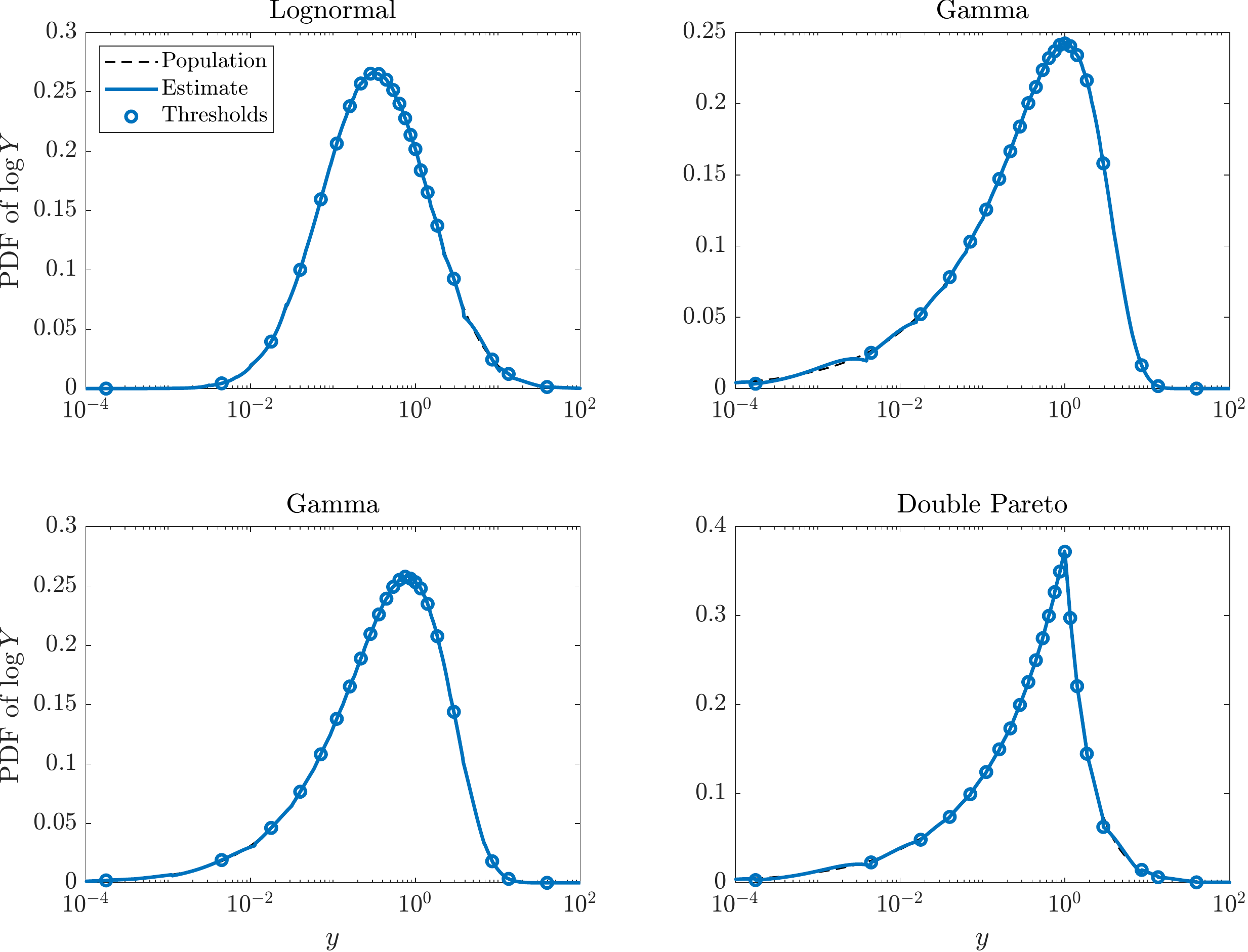}
\caption{PDF of $\log Y$ for each model.}\label{fig:sim}
\end{figure}

%%%%%%%%%%%%%%%%%%%%%%

\subsection{Top income shares}\label{subsec:sim_topshare}

Next, we estimate the top income shares, which correspond to the Lorenz curve flipped along the 45 degree line. 
There are many existing methods for estimating the Lorenz curve as discussed in Section \ref{sec:Intro}. 
We implement those proposed by \cite{KakwaniPodder1976} (henceforth KP) and \cite{VillasenorArnold1989} (henceforth VA), both of which have been used by the World Bank. 
In addition, we implement a more recently developed method by \cite*{Hajargasht2012} (henceforth HGBRC), which has been further extended by \cite{Chen2018} and \cite{HajargashtGriffiths2020}. 
The KP and VA methods impose some parametric assumptions on the Lorenz curve and essentially run linear regressions of the group mean ($\hat{y}_k$ in our notation) on some transformation of the proportion of each group ($\hat{q}_k$ in our notation). 
The HGBRC method imposes some parametric assumptions on the underlying density and constructs a generalized method of moments estimation.  
Regarding KP, we implement their Method III as described in their Section 4.
Regarding VA, we implement their method with $a=1$ and $d=0$ as described in their Section 4. 
Regarding HGBRC, we adopt their assumption of the generalized beta distribution of the second kind (GB2, \citealp{McDonald1984}) and the diagonal weighting matrix as proposed by \cite{Chotikapanich2007}.

Our data generating process is as follows. We suppose that the population distribution is double Pareto with parameters $\alpha=2.3$, $\beta=1.1$, and $M=1$ (normalization). There are two reasons for using the double Pareto distribution with these parameters. First, this distribution has been shown to fit the income distribution very well; see for instance \citet[Fig.~1(a)]{Toda2012JEBO}. Second, unlike other parametric distributions used in Figure \ref{fig:sim}, the double Pareto distribution admits a closed-form CDF and Lorenz curve as discussed in Appendix \ref{sec:dPLorenz}, which is convenient for numerical evaluation. Appendix \ref{sec:additional} considers other distributions.

We treat the population top $0.1, 1, 5,10,\dots,95,100$ percentiles as the observed thresholds ($K=22$) and compute the population top income shares. Next, we generate random samples with sizes $n=10^4, 10^5, 10^6$ from the population distribution\footnote{Since the logarithm of a double Pareto random variable is Laplace, which is double exponential, we can generate a double Pareto random variable using $\log (Y/M)=\frac{1}{\alpha}X_1-\frac{1}{\beta}X_2$, where $X_1,X_2$ are independent exponential random variables with parameter 1. Therefore $Y=MU_1^{-1/\alpha}U_2^{1/\beta}$, where $U_1,U_2$ are independent uniform random variables on $[0,1]$.} and record the proportions of observations and their average incomes within each group, which we treat as our data.

Implementing our proposed ME, the KP, the VA, and the HGBRC methods, we report their relative bias and relative root mean squared error (RMSE) for the income shares of the top $p_0$ fractile with $p_0 \in \set{0.001, 0.01, 0.05, 0.1, 0.2, \dots, 0.9}$. 
More specifically, let $s_0$ denote the true top income share and $\hat{s}_m$ the estimator in the $m$-th simulation draw with $m\in \set{1,\dots,M}$.   
We define the relative bias and RMSE by
\begin{align*}
\text{Relative Bias}&=\frac{1}{M}\sum_{m=1}^M(\hat{s}_m/s_0-1),\\
\text{Relative RMSE}&=\sqrt{\frac{1}{M}\sum_{m=1}^M(\hat{s}_m/s_0-1)^2},
\end{align*}
respectively. 
Table \ref{t:simu} presents the results based on $M=1{,}000$ simulations. 

\begin{table}[!htb]
\centering
${}$
\caption{Relative bias for top income shares.}\label{t:simu}
\resizebox{\textwidth}{!}{
\begin{tabular}{lrrrrrrrrrrrr}
\toprule

\multicolumn{13}{l}{\textbf{Relative Bias}}\\
$p_0$ & 0.001 & 0.01 & 0.05 & 0.1 & 0.2 & 0.3 & 0.4 & 0.5 & 0.6 & 0.7 & 0.8 & 0.9 \\
\midrule
$n$             & \multicolumn{12}{c}{\bf ME} \\
$10^4$ &0.007 &0.003 &0.001	 &0.001	 &0.001	 &0.000	 &0.000	  &0.000	 &0.000	 &0.000	 &0.000	 &0.000\\
$10^5$ &0.001 &0.000 &0.000	 &0.000	 &0.000	 &0.000	 &0.000	  &0.000	 &0.000	 &0.000	 &0.000	 &0.000\\
$10^6$ &0.000 &0.001 &0.000	 &0.000	 &0.000	 &0.000	 &0.000	  &0.000	 &0.000	 &0.000	 &0.000	 &0.000\\
\midrule
$n$            & \multicolumn{12}{c}{\bf KP} \\
$10^4$ &	-0.353	&-0.348	&0.025	&0.125	&0.101	&0.033	&-0.020	&-0.036	&-0.031	&-0.014	&0.002	&0.005\\
$10^5$ &	-0.349	&-0.343	&0.023	&0.123	&0.099	&0.032	&-0.020	&-0.037	&-0.031	&-0.014	&0.002	&0.005\\
$10^6$ &	-0.347	&-0.342	&0.023	&0.123	&0.099	&0.032	&-0.021	&-0.037	&-0.031	&-0.014	&0.002	&0.005\\
\midrule
$n$          & \multicolumn{12}{c}{\bf VA} \\
$10^4$ &	0.257	&0.052	&0.020	&0.029	&0.021	&-0.011	&-0.039	&-0.035	&-0.019	&-0.005	&0.002	&0.002\\
$10^5$ &	0.253	&0.049	&0.018	&0.027	&0.020	&-0.012	&-0.039	&-0.035	&-0.019	&-0.005	&0.002	&0.002\\
$10^6$ &	0.253	&0.049	&0.017	&0.027	&0.020	&-0.012	&-0.039	&-0.035	&-0.019	&-0.005	&0.002	&0.002\\
\midrule
$n$          & \multicolumn{12}{c}{\bf HGBRC} \\
$10^4$ &	-0.094	&-0.054	&-0.015	&0.003	&0.017	&0.009	&-0.007	&-0.007	&-0.004	&-0.001	&0.000	&0.000\\
$10^5$ &	-0.049	&-0.023	&0.001	&0.012	&0.019	&0.010	&-0.005	&-0.005	&-0.002	&-0.001	&0.000	&0.000\\
$10^6$ &	-0.039	&-0.019	&0.003	&0.012	&0.019	&0.010	&-0.005	&-0.005	&-0.002	&-0.001	&0.000	&0.000\\
\midrule
\multicolumn{13}{l}{\textbf{Relative RMSE}}\\
$p_0$ & 0.001 & 0.01 & 0.05 & 0.1 & 0.2 & 0.3 & 0.4 & 0.5 & 0.6 & 0.7 & 0.8 & 0.9 \\
\midrule
$n$             & \multicolumn{12}{c}{\bf ME} \\
$10^4$ &	0.440	&0.133	&0.053	&0.033	&0.019	&0.012	&0.009	&0.006	&0.004	&0.002	&0.001	&0.000\\
$10^5$ &	0.136	&0.042	&0.017	&0.010	&0.006	&0.004	&0.003	&0.002	&0.001	&0.001	&0.000	&0.000\\
$10^6$ &	0.046	&0.014	&0.005	&0.003	&0.002	&0.001	&0.001	&0.001	&0.000	&0.000	&0.000	&0.000\\
\midrule
$n$             & \multicolumn{12}{c}{\bf KP} \\
$10^4$ &	0.427	&0.370	&0.062	&0.130	&0.102	&0.035	&0.021	&0.037	&0.031	&0.014	&0.002	&0.005\\
$10^5$ &	0.357	&0.346	&0.029	&0.123	&0.099	&0.032	&0.021	&0.037	&0.031	&0.014	&0.002	&0.005\\
$10^6$ &	0.348	&0.343	&0.024	&0.123	&0.099	&0.032	&0.021	&0.037	&0.031	&0.014	&0.002	&0.005\\
\midrule
$n$        & \multicolumn{12}{c}{\bf VA} \\
$10^4$ &	0.268	&0.079	&0.051	&0.048	&0.034	&0.022	&0.041	&0.036	&0.020	&0.005	&0.002	&0.002\\
$10^5$ &	0.254	&0.052	&0.023	&0.030	&0.022	&0.013	&0.039	&0.035	&0.019	&0.005	&0.002	&0.002\\
$10^6$ &	0.253	&0.049	&0.018	&0.027	&0.020	&0.012	&0.039	&0.035	&0.019	&0.005	&0.002	&0.002\\
\midrule
$n$          & \multicolumn{12}{c}{\bf HGBRC} \\
$10^4$ &	0.223	&0.137	&0.071	&0.044	&0.027	&0.015	&0.011	&0.009	&0.005	&0.002	&0.001	&0.000\\
$10^5$ &	0.093	&0.049	&0.024	&0.019	&0.021	&0.011	&0.006	&0.005	&0.003	&0.001	&0.000	&0.000\\
$10^6$ &	0.054	&0.026	&0.011	&0.015	&0.020	&0.011	&0.005	&0.005	&0.002	&0.001	&0.000	&0.000\\

\bottomrule
\end{tabular}
}
\caption*{\footnotesize Note: Relative Bias and RMSE for the income shares of the top $p_0$ fractile. ME denotes our proposed method. KP, VA, and HGBRC denote the methods proposed respectively by \cite{KakwaniPodder1976}, \cite{VillasenorArnold1989}, and \cite{Hajargasht2012}. See the main text for the data generating process. The results are based on 1,000 simulation draws.}
\end{table}

The findings can be summarized as follows. 
First, our proposed ME method performs very well in terms of both bias and RMSE. 
They decrease as $n$ increases and are smaller than those of the other three methods for most of the $(n,p_0)$ combinations, especially the bias. 
Second, the KP and the VA methods both impose some parametric assumptions on the Lorenz curve and hence implicitly on the underlying density. 
In particular, the VA method imposes that the Lorenz curve is a part of an ellipse. 
This assumption implies that the underlying density $f(y)$ (after a location- and scale-transformation) is proportional to $(1+y^2/2)^{-3/2}$, which is the Student $t$ distribution with two degrees of freedom \citep[Theorem 4]{VillasenorArnold1989}.
The KP method introduces a new coordinate system and imposes another parametric form on the Lorenz curve. 
The implied density is still parametric but does not have a closed-form expression.  
The HGBRC method assumes the GB2 density, which has Pareo upper and lower tails. Therefore, its performance is substantially better than those of KP and VA. 
In summary, these and any other parametric assumptions could lead to large bias and RMSE caused by misspecification, which do not decrease with $n$.

\subsection{Comparison to kernel estimator}\label{subsec:sim_kernel}

Finally, we compare our proposed ME estimator with the nonparametric kernel estimator proposed by \citet{BlowerKelsall2002}. 
Given a bandwidth $h$, define $K_h(u)$ as the PDF of the normal distribution with mean zero and variance $h^2$, that is, 
\begin{equation*}
K_{h}(u) =\frac{1}{\sqrt{2\pi}h}\exp \left( -\frac{u^2}{2h^2}\right). 
\end{equation*}
We implement \citet[][eq.(1.4)]{BlowerKelsall2002} by constructing the density estimator
\begin{equation*}
\hat{f}_\mathrm{BK}(y) =\sum_{k=2}^{K}\hat{q}_{k}\frac{\int_{I_{k}}K_{h}(s-y) \hat{f}_{0}(s) \diff s}{\int_{I_{k}}\hat{f}_{0}(s) \diff s},
\end{equation*}
where $\hat{f}_{0}(s)$ is the histogram estimator
\begin{equation*}
\hat{f}_{0}(y) =\sum_{k=2}^{K}\frac{\hat{q}_{k}}{t_{k-1}-t_{k}}\id_{I_k}(y).
\end{equation*}
Using the fact that $\int_{I_{k}}\hat{f}_{0}(s) \diff s=\hat{q}_{k}$ and $K_{h}$ is the normal density, we can simplify $\hat{f}_\mathrm{BK}$ as follows:
\begin{align*}
\hat{f}_\mathrm{BK}(y)  &= \sum_{k=2}^{K}\int_{I_{k}}K_{h}(s-y) \hat{f}_{0}(s) \diff s =\sum_{k=2}^{K}\frac{\hat{q}_{k}}{t_{k-1}-t_{k}}\int_{t_{k}}^{t_{k-1}}K_{h}(s-y) \diff s \\
&=\sum_{k=2}^{K}\frac{\hat{q}_{k}}{t_{k-1}-t_{k}}\left( \Phi \left( \frac{t_{k-1}-y}{h}\right) -\Phi \left( \frac{t_{k}-y}{h}\right) \right) ,
\end{align*}
where $\Phi$ denotes the CDF of the standard normal distribution.

\citet{BlowerKelsall2002} do not derive any asymptotic properties of this estimator nor theoretical requirements on the choice of the bandwidth. We implement a variety of choices of $h$ to examine its finite sample performance. 
Specifically, we use the rule-of-thumb choice $h=c\hat{\sigma}n^{-1/5}$, where $\hat{\sigma}$ is the sample standard deviation based on individual observations (which is in principle infeasible given the tabulated data) and $c\in\set{0.1, 0.5, 1.0, 1.5}$ is some constant.

Figure \ref{fig:kernel} presents the relative RMSE for the density estimators when the data generating process is double Pareto, lognormal, and gamma as in Section \ref{subsec:sim_topshare} with sample size $n=10^4,10^5,10^6$. 
Although these figures are not necessarily easy to read, the RMSEs for the ME (BK) estimator are indicated with solid (dashed) lines. As is clear from this figure, the RMSEs for the ME estimator is generally closest to the horizontal axis uniformly across quantiles, so the performance of our proposed ME method is outstanding. 
In addition, we also implement the kernel estimator studied by \citet{Reyes2016}. 
Its performance is substantially worse than that proposed by \citet{BlowerKelsall2002} and hence not reported. 

\begin{figure}[!htb]
\centering
\includegraphics[width=\linewidth]{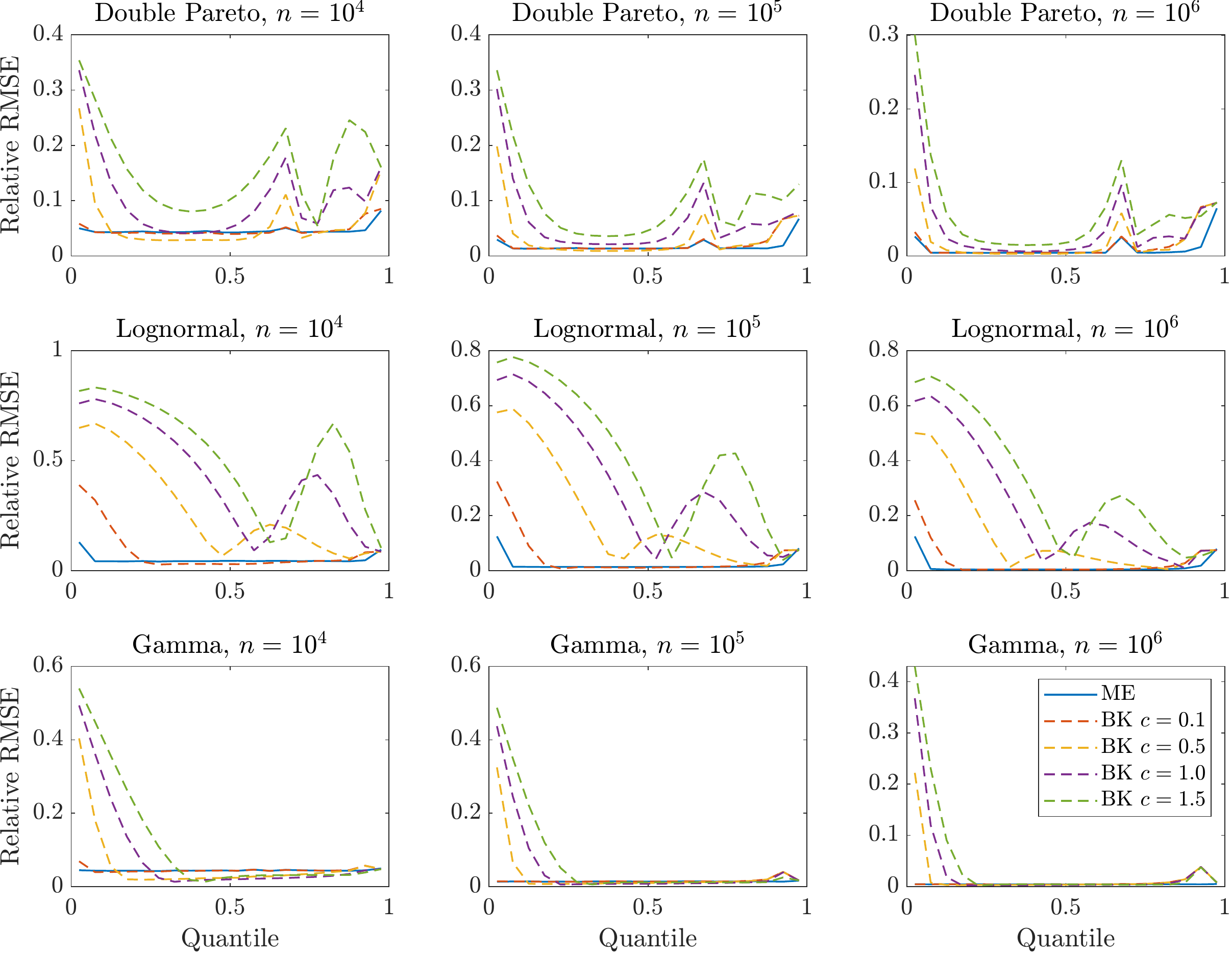}
\caption{Relative RMSE of density estimators.}\label{fig:kernel}
\caption*{\footnotesize Note: The figure presents the relative RMSE for the density estimators when the data generating process is double Pareto, lognormal, and gamma with sample size $n=10^4,10^5,10^6$. ME denotes our proposed maximum entropy method, whereas BK denotes the kernel estimator of \cite{BlowerKelsall2002} with various bandwidth parameters $c$.}
\end{figure}

%%%%%%%%%%%%%%%%%%%%%%
\section{Income distribution in the United States}
\label{sec:Income}
%%%%%%%%%%%%%%%%%%%%%%

We consider two empirical applications of our method. First, we estimate the distribution of U.S. income distribution for particular years. Second, we estimate the top income shares (including mid-sample) over the past century.

\subsection{Income distribution in 1946 and 2019}

We estimate the distribution of U.S. adjusted gross income (AGI) in 1946 and 2019. We choose 2019 because it is the most recent year for which data is available. Before World War II, because only a small fraction of the population filed for taxes, the tax returns data is not representative for the population.\footnote{The fraction of tax filers among potential tax units has been stable at around 80--90\% postwar but in the range of 1--20\% before 1940; see the discussion in \cite{PikettySaez2003}.} For this reason, we choose 1946 because it is one of the earliest years for which the tax returns data is representative for the population. Note that unlike in recent years, the tabulated summary data set is almost the only publicly available income data set in early years such as 1946.

To make the results comparable across years, we measure income in 2019 dollars by adjusting with the Consumer Price Index (CPI). The number of income groups is $K=48$ for 1946 and $K=18$ for 2019. Figure \ref{fig:logincomePDF} shows the ME density estimates $\hat{f}(\e^x)\e^x$ of log income $x=\log y$ in a semi-log scale. 
 
We can summarize the findings as follows. 
First, for each year the log income distribution is bell-shaped but slightly asymmetric. Second, observe that although $\hat{f}$ is piecewise exponential, it is not necessarily continuous at the bin thresholds as can be seen from the spikes in the 1946 density. Third, the 2019 density is more spread-out than 1946, which suggests that income inequality has increased.  Finally, Figure \ref{fig:loglog} shows the tail probability $1-\hat{F}(y)$ in a log-log scale, which is continuous. The fact that the 2019 tail probability is higher than 1946 implies that the 2019 (real) income distribution first-order stochastically dominates the 1946 one, possibly due to economic growth. The graphs also show a straight-line pattern for high incomes, which is consistent with a Pareto upper tail documented elsewhere; see for instance \cite{deVriesToda2022RIW} and the references therein. Because the slope is steeper for 1946 than in 2019, the income Pareto exponent is smaller (top income inequality is higher) in 2019.

\begin{figure}[!htb]
\centering
\begin{subfigure}{0.48\linewidth}
\includegraphics[width=\linewidth]{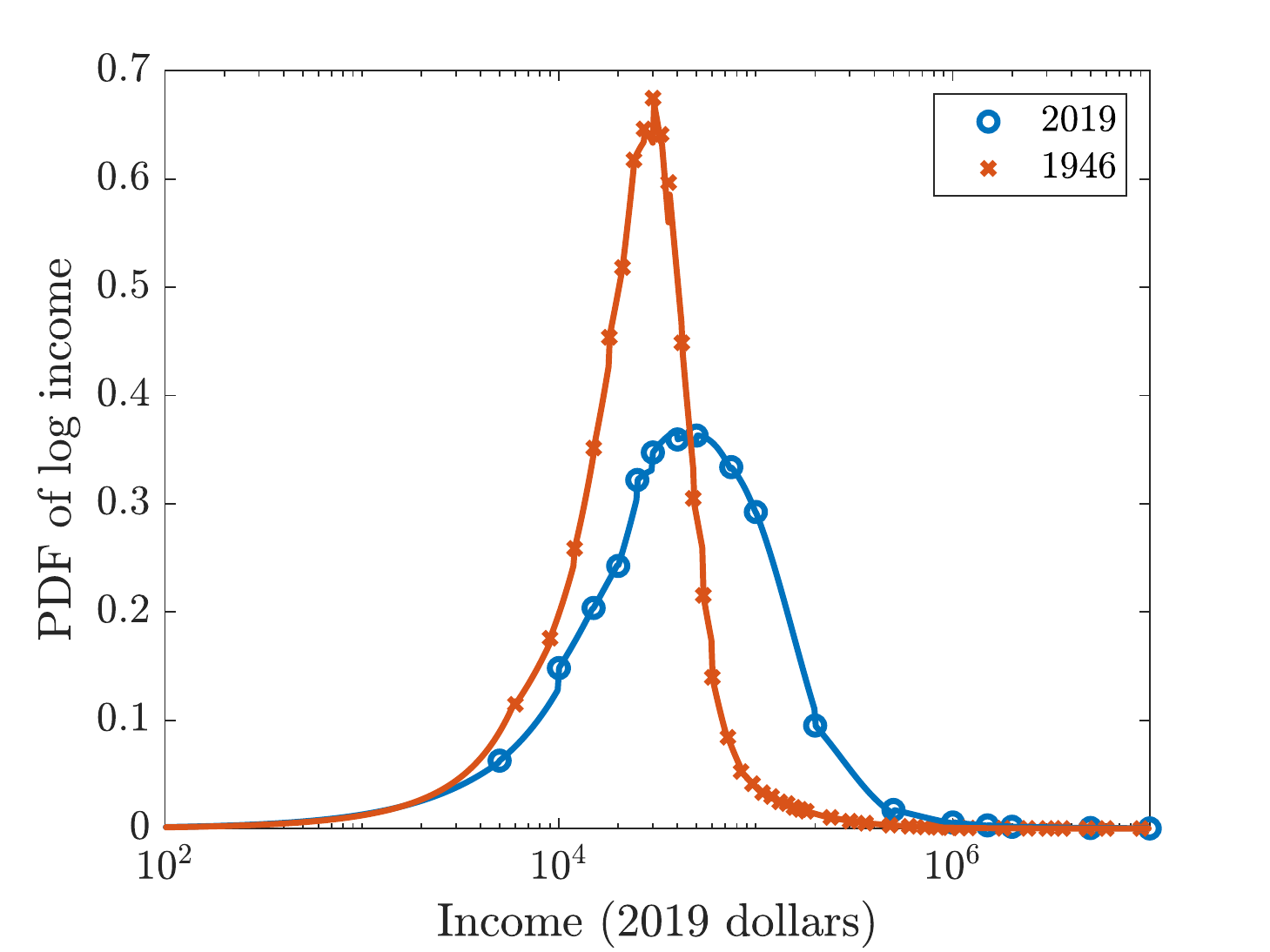}
\caption{Probability density of log income.}\label{fig:logincomePDF}
\end{subfigure}
\begin{subfigure}{0.48\linewidth}
\includegraphics[width=\linewidth]{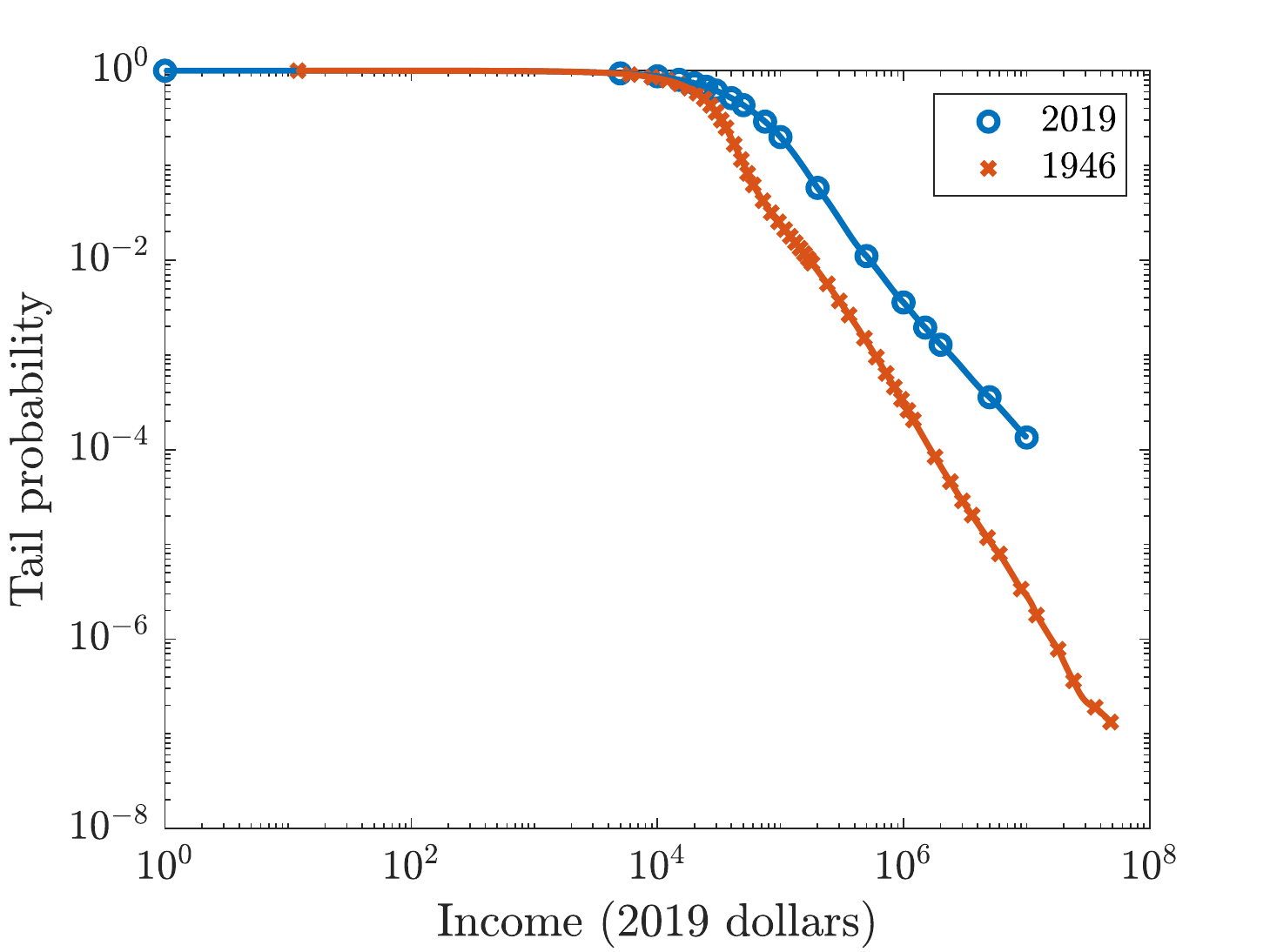}
\caption{Log-log plot of tail probability.}\label{fig:loglog}
\end{subfigure}
\caption{U.S. income distributions in 1946 and 2019.}
\end{figure}

Because the ME density $\hat{f}$ is piecewise exponential, which is analytically tractable, it is straightforward to compute statistics such as top income shares; see Section \ref{subsec:discussion}. Figure \ref{fig:topshare} shows the top income shares, both in original and log-log scales. Figure \ref{fig:topshare_original} (original scale) is essentially the Lorenz curve flipped along the 45 degree line. The fact that the 2019 curve is above the 1946 one suggests that income inequality has increased. The straight-line pattern in log-log scale (Figure \ref{fig:topshare_loglog}) is consistent with a Pareto upper tail.

\begin{figure}[!htb]
\centering
\begin{subfigure}{0.48\linewidth}
\includegraphics[width=\linewidth]{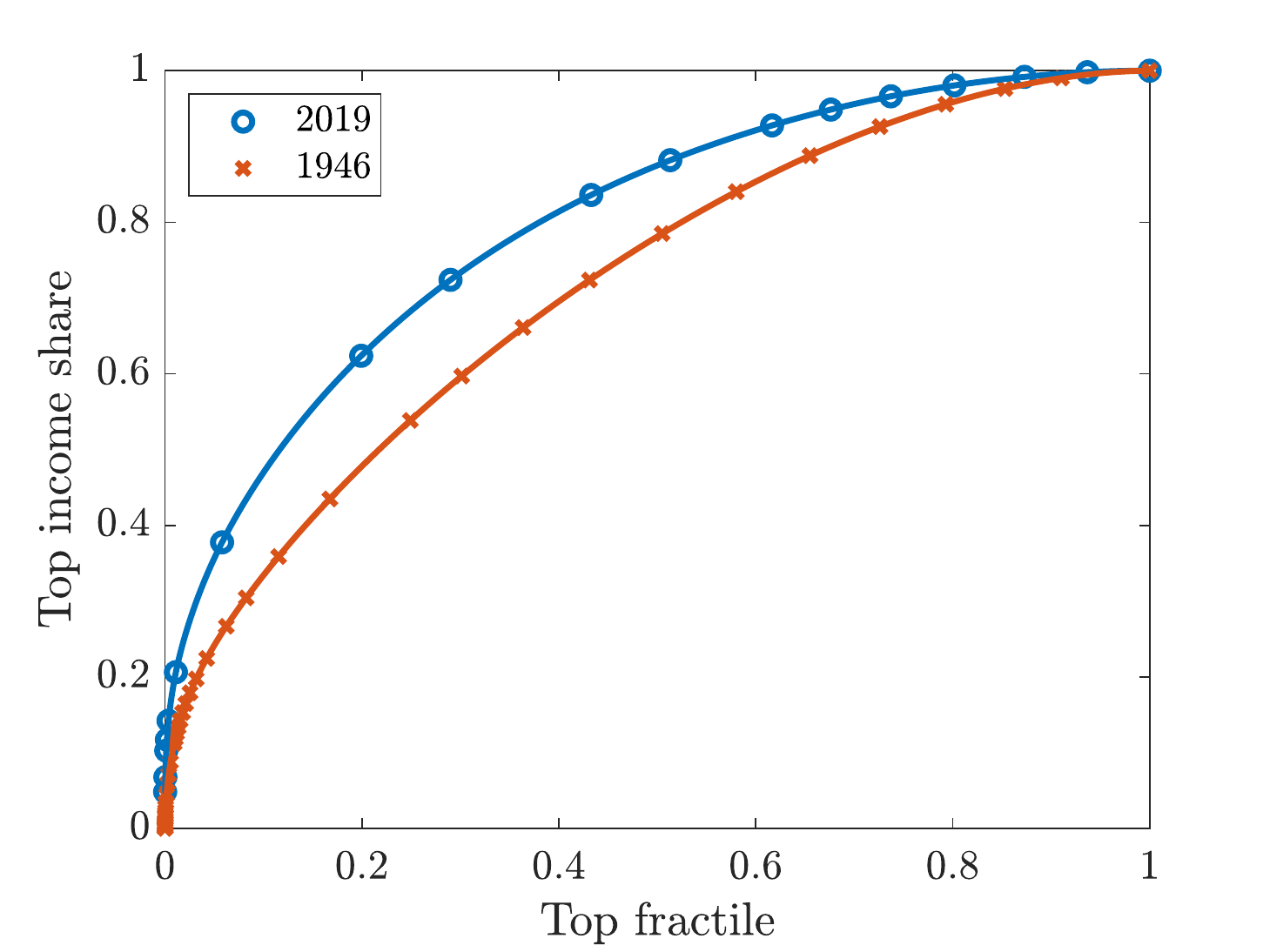}
\caption{Original scale.}\label{fig:topshare_original}
\end{subfigure}
\begin{subfigure}{0.48\linewidth}
\includegraphics[width=\linewidth]{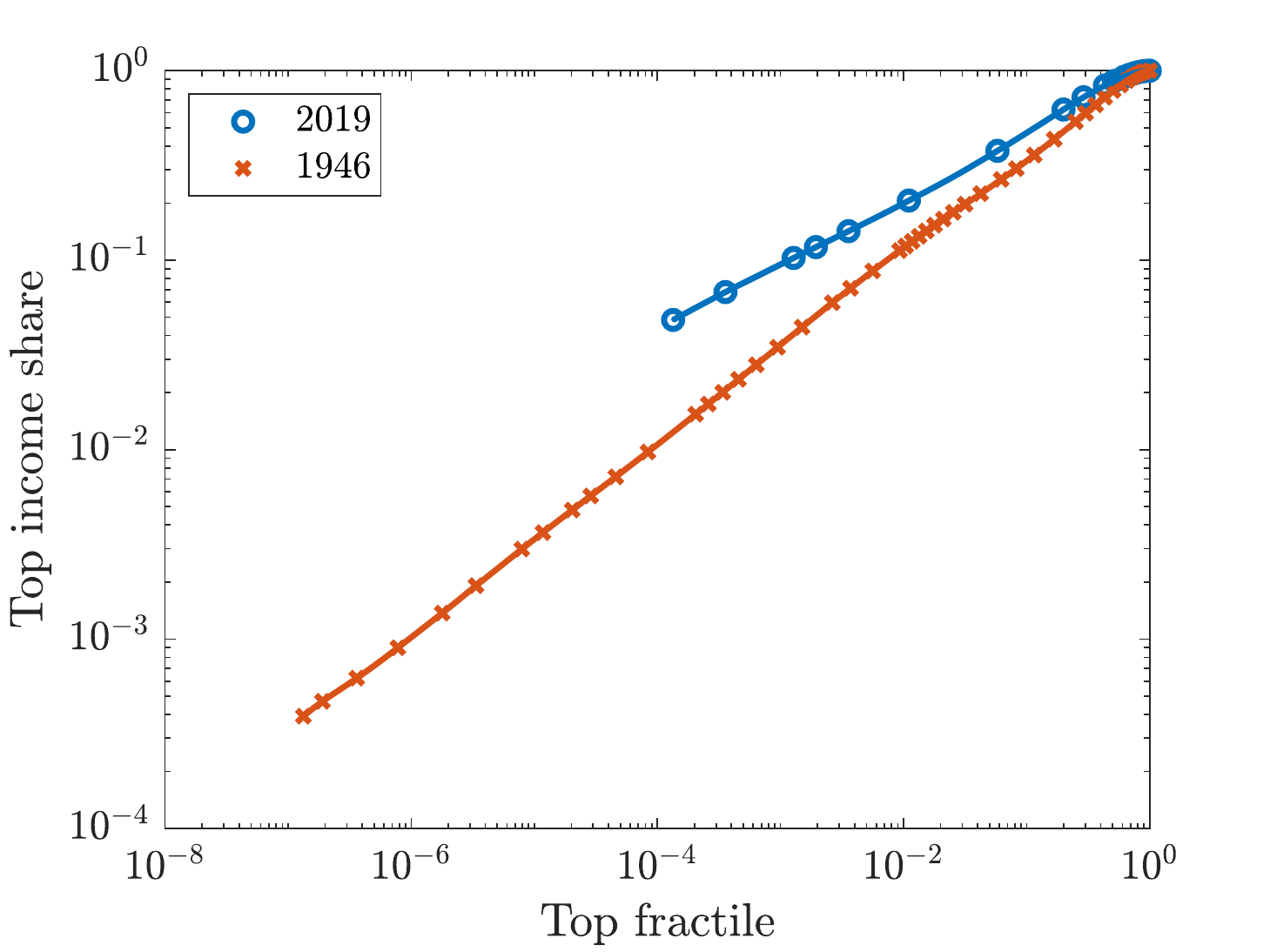}
\caption{Log-log scale.}\label{fig:topshare_loglog}
\end{subfigure}
\caption{U.S. top income shares in 1946 and 2019.}\label{fig:topshare}
\end{figure}

\subsection{Top income shares over the century}

Finally, we apply the proposed method to estimate the top $p$ fractile income share for various values of $p\in [0,1]$. To construct the top income shares, we use the following approach. First, we collect the tabulated summaries of income similar to Table \ref{t:IRS2019} for each year from the IRS \emph{Statistics of Income}.\footnote{See the appendix in \cite{LeeSasakiTodaWangExponents} for specific details.} These tables contain information on the number of tax units (an individual or a married couple with dependents if any) and their total income within each income group. As these tables contain only tax filers, we complement them with the total number of potential tax units and total income estimated by \cite{PikettySaez2003}.\footnote{We obtain the total number of tax units from the spreadsheet \url{https://eml.berkeley.edu/~saez/TabFig2018.xls}, Table A0, Column B, and total income from Column I.} We suppose that non-filers are low income households and thus do not affect the calculation of the top $p$ fractile income share if $p$ is small enough. Because the fraction of tax filers among potential tax units exceeds 0.1 (0.8) since 1936 (1945), we construct the top $p$ fractile income share for $p\in \set{0.01,0.05,0.1}$ since 1936 and also for for $p\in \set{0.2, 0.4, 0.6, 0.8}$ since 1945. Figure \ref{fig:topshare_1-80} shows the results.

\begin{figure}[htb!]
\centering
\includegraphics[width=0.7\linewidth]{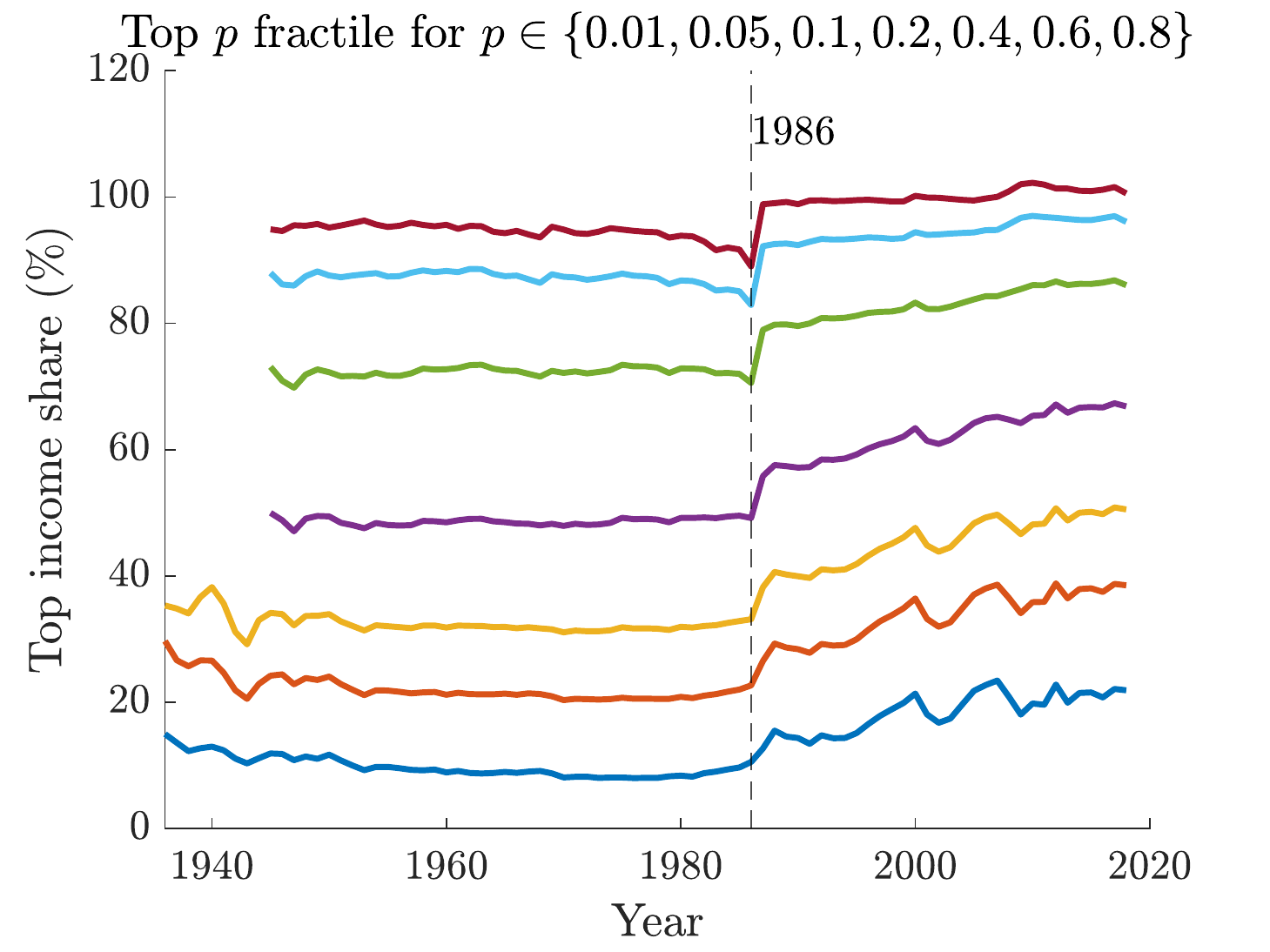}
\caption{Top income shares in U.S., 1936--2018.}\label{fig:topshare_1-80}
\caption*{\footnotesize Note: From bottom to top, the figure plots the top $p$ fractile income share for $p\in\set{0.01,0.05,0.1,0.2,0.4,0.6,0.8}$. Mid-sample fractiles ($p>0.1$) are omitted pre-1945 because the fraction of tax filers is too small.}
\end{figure}

The top 1\%, 5\%, 10\% income shares exhibit an inverse U-shaped pattern, which is well known. To the best of our knowledge, the top income shares for mid-sample fractiles (\eg, $p=0.2, 0.4, 0.6, 0.8$) have not been reported in the previous literature. We find that the mid-sample top income shares exhibit an abrupt upward jump between 1986 and 1987. This could be due to the Tax Reform Act of 1986, which significantly altered the treatment of capital gains income.\footnote{\label{fn:cg}According to \citet[p.~40]{PikettySaez2001WP}, the fraction of capital gains income included in AGI was 100\% until 1933, 70\% in 1934--1937, 60\% in 1938--1941, 50\% in 1942--1978, 40\% in 1979--1986, and 100\% since 1987. Because high income earners tend to hold more financial asset (and generate more capital gains), the large exclusion of capital gains in 1934--1986 likely causes the top income shares to be biased downwards.}

We next compare our top income shares to those constructed by \cite{PikettySaez2003}.\footnote{We obtain the top income shares from Table A3 in the spreadsheet \url{https://eml.berkeley.edu/~saez/TabFig2018.xls}. \citet[Appendix B, Section 1.1, pp.~592--599]{Piketty2001book} provides the details of the method for constructing these top income shares. Since it is written in French, we describe the method in Appendix \ref{sec:piketty} for the convenience of the readers.} Figure \ref{fig:topshare_1-10} shows the top 1\%, 5\%, 10\% income shares constructed in two ways. We find that post-1986, our top income shares are nearly identical to those from \cite{PikettySaez2003}, even though our method is nonparametric while their method is parametric (assuming a Pareto upper tail). This is likely because the upper tail of the income distribution can be well approximated by the Pareto distribution. However, there are large discrepancies between the two series pre-1986 because we have used the raw AGI without adjusting for the excluded capital gains income discussed in Footnote \ref{fn:cg}.

\begin{figure}[!htb]
\centering
\includegraphics[width=0.7\linewidth]{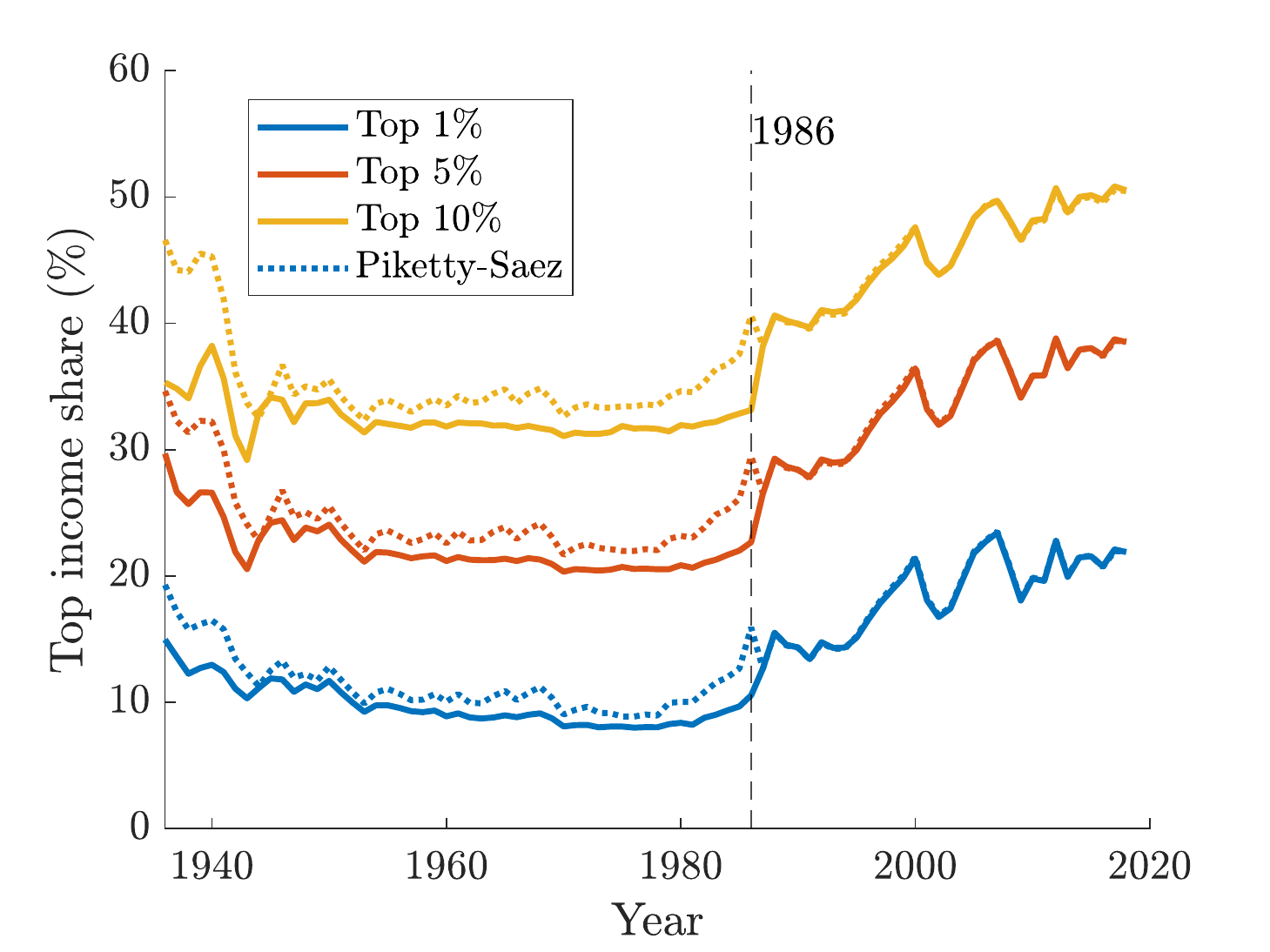}
\caption{Comparison to \cite{PikettySaez2003} series.}\label{fig:topshare_1-10}
\end{figure}

%%%%%%%%%%%%%%%%%%%%%%
\section{Concluding remarks}
\label{sec:Conclusion}
%%%%%%%%%%%%%%%%%%%%%%

Existing standard nonparametric estimators of density and cumulative distribution functions require individual-level data.
Even when individual-level information is difficult to access due to confidentiality concerns, tabulated summaries of such data are often publicly available.
Administrative data of income are the leading examples.
In this paper, we propose a novel method of maximum entropy density estimation from tabulated summary data and establish the strong uniform consistency of the density and cumulative distribution estimators.
This method enjoys many desirable properties.
First, the estimator is piecewise exponential, which is analytically tractable.
Using its functional form, it is straightforward to compute statistics such as top income shares.
Second, and more importantly, our estimator is free from tuning parameters unlike existing kernel-based methods, which is attractive in practice.
This feature provides a complete theoretical justification that our proposed estimator works in practice, unlike the existing kernel-based estimators for which the theory does not formally account for the effects of bandwidth choice in practice.

\appendix
%%%%%%%%%%%%%%%%%%%%%%
\section*{Appendix}\label{sec:Appendix}
%%%%%%%%%%%%%%%%%%%%%%
Appenidx \ref{sec:ME_unknown_thresholds} provides a smoothed version of the proposed ME estimator. 
Appendix \ref{sec:proof} contains mathematical proofs of the strong uniform consistency. 
Appendix \ref{sec:numerical} solves some numerical issues in constructing the proposed estimator.
Appendix \ref{sec:dPLorenz} computes the Lorenz curve for the double Pareto distribution.
Appendix \ref{sec:piketty} describes the details of \cite{Piketty2003}'s Pareto interpolation method, which was originally written in French. 
Appendix \ref{sec:additional} contains additional simulation studies.

%%%%%%%%%%%%%%%%%%%%%%
\section{Smoothed ME estimation}\label{sec:ME_unknown_thresholds}
%%%%%%%%%%%%%%%%%%%%%%

Our estimator described in Section \ref{subsec:ME} in the main text is generally discontinuous at the bin thresholds $t_1,\dots,t_K$. 
In practice, an \emph{ad hoc} yet simple method to smooth the estimator is to replace it with some local polynomial function within small neighborhoods of $t_1,\dots,t_K$. 
As a more systematic alternative, we can slightly shift the bin boundaries so that the implied ME estimator becomes continuous. 
To that end, we can treat $t_1,\dots,t_K$ as arguments and further minimize the Kullback-Leibler divergence \eqref{eq:Jstar} over $t_1,\dots,t_K$ as follows. 
Because the minimum threshold $t_K$ is unidentified, we set it to an arbitrary value $t_K<y_K$, say, $t_K=0$. 
Define the set of admissible thresholds by
\begin{align}
T &= \set{(t_1,\dots,t_{K-1})\in \R^{K-1}: (\forall k) t_k<y_k<t_{k-1}} \notag \\
&=(y_2,y_1)\times (y_3,y_2) \times\dots\times(y_K,y_{K-1}), \label{eq:T}
\end{align}
which is a nonempty open interval in $\R^{K-1}$. 
The following theorem shows that $J^*$ in \eqref{eq:Jstar} always achieves a unique minimum on $T$, which is continuous. 

\begin{thm}\label{thm:ME_threshold}
Fix $t_K<y_K$ and let $J^*:T\to \R$ be defined by \eqref{eq:Jstar}. Then $J^*$ is strictly convex and achieves a unique minimum $t^*\in T$. Furthermore, the corresponding maximum entropy density $f^*$ in \eqref{eq:fstar} with $t=t^*$ is continuous on $[t_K,\infty)$.
\end{thm}

The proof is deferred to Appendix \ref{sec:proof}.

%%%%%%%%%%%%%%%%%%%%%%
\section{Proofs}\label{sec:proof}
%%%%%%%%%%%%%%%%%%%%%%

\begin{proof}[Proof of Proposition \ref{prop:ME_density}]
To simplify notation, we suppress the dependence of $J_k$, $\lambda_k^*$, etc.\ on $t$. We use Fenchel duality \citep{borwein-lewis1991} to solve the maximum entropy problem. The Lagrangian of the maximum entropy problem subject to the moment conditions \eqref{eq:mcond} is
\begin{align*}
L(g,\mu,\lambda)&=\int g\log g\diff y+\sum_{k=1}^K\left(\mu_k\left(q_k-\int \id_{I_k}g\diff y\right)+\lambda_k\left(q_ky_k-\int \id_{I_k}yg\diff y\right)\right)\\
&=\sum_{k=1}^Kq_k(\mu_k+y_k\lambda_k)+\int\left(g\log g-\sum_{k=1}^K(\mu_k\id_{I_k}g+\lambda_k\id_{I_k}yg)\right)\diff y,
\end{align*}
where $\mu_k$ and $\lambda_k$ are the Lagrange multipliers corresponding to the moment conditions \eqref{eq:mcond_prob} and \eqref{eq:mcond_mean}. The dual objective function is the minimum of $L$ over unconstrained $g\in L_+^1(I)$. Taking the first order condition (G\^ateaux derivative) pointwise, for $y\in I_k$ we obtain
\begin{equation}
0=\log g+1-(\mu_k+\lambda_ky)\iff g(y)=\e^{\mu_k+\lambda_ky-1}. \label{eq:foc_g}
\end{equation}
Hence the dual objective function becomes
\begin{equation*}
J(\mu,\lambda)\coloneqq \min_{g\in L_+^1(I)}L(g,\mu,\lambda)=\sum_{k=1}^Kq_k(\mu_k+y_k\lambda_k)-\sum_{k=1}^K\int_{I_k}\e^{\mu_k+\lambda_ky-1}\diff y.
\end{equation*}
The dual problem maximizes $J$ with respect to $\mu,\lambda$, which is a concave maximization problem. The first order condition with respect to $\mu_k$ is
\begin{equation}
0=\frac{\partial J}{\partial \mu_k}=q_k-\int_{I_k}\e^{\mu_k+\lambda_ky-1}\diff y\iff \mu_k=1+\log q_k-\log \left(\int_{I_k}\e^{\lambda_ky}\diff y\right).\label{eq:muk}
\end{equation}
Then (with a slight abuse of notation) the objective function becomes
\begin{align}
J(\lambda)\coloneqq \max_{\mu\in \R^K} J(\mu,\lambda)&=\sum_{k=1}^K\left(q_ky_k\lambda_k+q_k\log q_k-q_k\log\left(\int_{I_k}\e^{\lambda_ky}\diff y\right)\right) \notag \\
&=\sum_{k=1}^Kq_k\left(y_k\lambda_k+\log q_k-\log\left(\int_{I_k}\e^{\lambda_ky}\diff y\right)\right) \notag \\
&=\sum_{k=1}^Kq_k(J_k(\lambda_k)+\log q_k), \label{eq:Jlambda}
\end{align}
where $J_k$ is given by \eqref{eq:Jk}. Since $J$ is additively separable, it suffices to maximize $J_k$ for each $k$.

Let us now show that $J_k(\lambda)$ achieves a unique maximum over $\lambda\in \R$. It is straightforward to show the strict concavity of $J_k$ by applying H\"older's inequality to the function $\lambda\mapsto \log\left(\int_{I_k}\e^{\lambda y}\diff y\right)$. As $\lambda\to \pm \infty$, it follows from \eqref{eq:Jk} that
\begin{align*}
J_k(\lambda)&=y_k\lambda-\log\left(\frac{\e^{\lambda t_{k-1}}-\e^{\lambda t_k}}{\lambda}\right)=\log\left(\frac{\lambda\e^{\lambda y_k}}{\e^{\lambda t_{k-1}}-\e^{\lambda t_k}}\right)\\
&=\log\left(\frac{\lambda}{\e^{\lambda (t_{k-1}-y_k)}-\e^{\lambda (t_k-y_k)}}\right)\to -\infty
\end{align*}
because $t_k<y_k<t_{k-1}$. Since $J_k$ is continuous and strictly concave on its domain, it achieves a unique maximum $\lambda_k\in \R$. When $k=1$, we have $J_1(\lambda)=-\infty$ for $\lambda\le 0$, and analytically maximizing $J_1$ for $\lambda>0$, we obtain the expression for $\lambda_1^*$ in \eqref{eq:lambdak_sign}. Differentiating $J_k$ under the integral sign, we obtain
\begin{equation*}
J_k'(\lambda)=y_k-\frac{\int_{I_k}y\e^{\lambda y}\diff y}{\int_{I_k}\e^{\lambda y}\diff y}\\
\implies J_k'(0)=y_k-\frac{\int_{I_k}y\diff y}{\int_{I_k}\diff y}=y_k-\frac{t_k+t_{k-1}}2,
\end{equation*}
so $\lambda_k^*$ can be signed as in \eqref{eq:lambdak_sign}.

If $\lambda_k^*\neq 0$, using \eqref{eq:foc_g} and \eqref{eq:muk} and suppressing the asterisks, we obtain the solution
\begin{equation*}
f^*(y)=g(y)=\e^{\mu_k+\lambda_ky-1}=\frac{q_k\lambda_k\e^{\lambda_k y}}{\e^{\lambda_kt_{k-1}}-\e^{\lambda_kt_k}},
\end{equation*}
which is \eqref{eq:fstar} and is exponential in $y$. The derivation for $\lambda_k^*=0$ is similar. The minimum value \eqref{eq:Jstar} follows from \eqref{eq:Jlambda}.
\end{proof}

\begin{lem}\label{lem:phi}
Define the function $\phi: (0,\infty)\to \R$ by
\begin{equation}
\phi(x)=\frac{\cosh x}{\sinh x}-\frac{1}{x},\label{eq:phi}
\end{equation}
where $\cosh x=\frac{\e^x+\e^{-x}}{2}$ and $\sinh x=\frac{\e^x-\e^{-x}}{2}$. Then $\phi'(x)>0$, $\phi''(x)<0$, $\phi(0+)=0$, $\phi(\infty)=1$, and $\phi'(0+)=1/3$.
\end{lem}

\begin{proof}
To simplify notation, let $c(x)=\cosh x$ and $s(x)=\sinh x$. $\phi(\infty)=1$ is trivial since $c(x)\sim \e^x/2\sim s(x)$ as $x\to \infty$. Using the definition of $c$ and $s$, we have $c'=s$, $s'=c$, and $c^2-s^2=1$. Taking the derivative of $\phi$, we have
\begin{equation*}
\phi'(x)=\frac{c's-cs'}{s^2}+\frac{1}{x^2}=\frac{s^2-c^2}{s^2}+\frac{1}{x^2}=-\frac{1}{s^2}+\frac{1}{x^2}.
\end{equation*}
Noting that $s(x)=x+x^3/6+o(x^3)$ as $x\to 0$ by Taylor's theorem, a straightforward calculation yields $\phi'(0+)=1/3$.

To show $\phi'>0$, it suffices to show $s>x$. To this end, define $g(x)=s(x)-x$. Then $g(0)=0$ and
\begin{equation*}
g'(x)=s'(x)-1=c(x)-1=\frac{\e^x+\e^{-x}}{2}-1=\frac{(\e^{x/2}-\e^{-x/2})^2}{2}>0
\end{equation*}
for $x>0$, so $g(x)>0$ for $x>0$. This shows $\phi'(x)>0$. Taking the derivative once again, we have
\begin{equation*}
\phi''(x)=\frac{2s'}{s^3}-\frac{2}{x^3}=2\left(\frac{c}{s^3}-\frac{1}{x^3}\right).
\end{equation*}
Therefore to show $\phi''<0$, it suffices to show $s^3/c>x^3$. To this end, define $h(x)=s(x)^3/c(x)-x^3$. Then
\begin{align*}
h'(x)&=\frac{3s^2s'c-s^3c'}{c^2}-3x^2=\frac{3s^2c^2-s^4}{c^2}-3x^2\\
&=3(c^2-1)-\frac{(c^2-1)^2}{c^2}-3x^2=2c^2-1-\frac{1}{c^2}-3x^2,\\
h''(x)&=4cc'+\frac{2c'}{c^3}-6x=4sc+\frac{2s}{c^3}-6x,\\
h'''(x)&=4s'c+4sc'+\frac{2s'}{c^3}-3\frac{2sc'}{c^4}-6=4c^2+4s^2+\frac{2}{c^2}-\frac{6s^2}{c^4}-6\\
&=8c^2-10-\frac{4}{c^2}+\frac{6}{c^4}=\frac{2}{c^4}(c^2-1)^2(4c^2+3)>0
\end{align*}
for $x>0$, where the last inequality follows from $c>1$ for $x>0$. Noting that $h'''(0)=h''(0)=h'(0)=h(0)=0$, it follows that $h(x)>0$ for $x>0$, implying $s^3/c>x^3$ and hence $\phi''<0$.
\end{proof}

Below, extend the domain of $\phi$ in \eqref{eq:phi} to the entire real line by setting $\phi(0)=0$ and $\phi(x)=-\phi(-x)$ for $x<0$. Then $\phi:\R\to (-1,1)$, $\phi$ is strictly increasing, and $\phi$ is concave (convex) for $x>0$ ($x<0$). Figure \ref{fig:phi_graph} shows its graph.

\begin{figure}[!htb]
\centering
\includegraphics[width=0.7\linewidth]{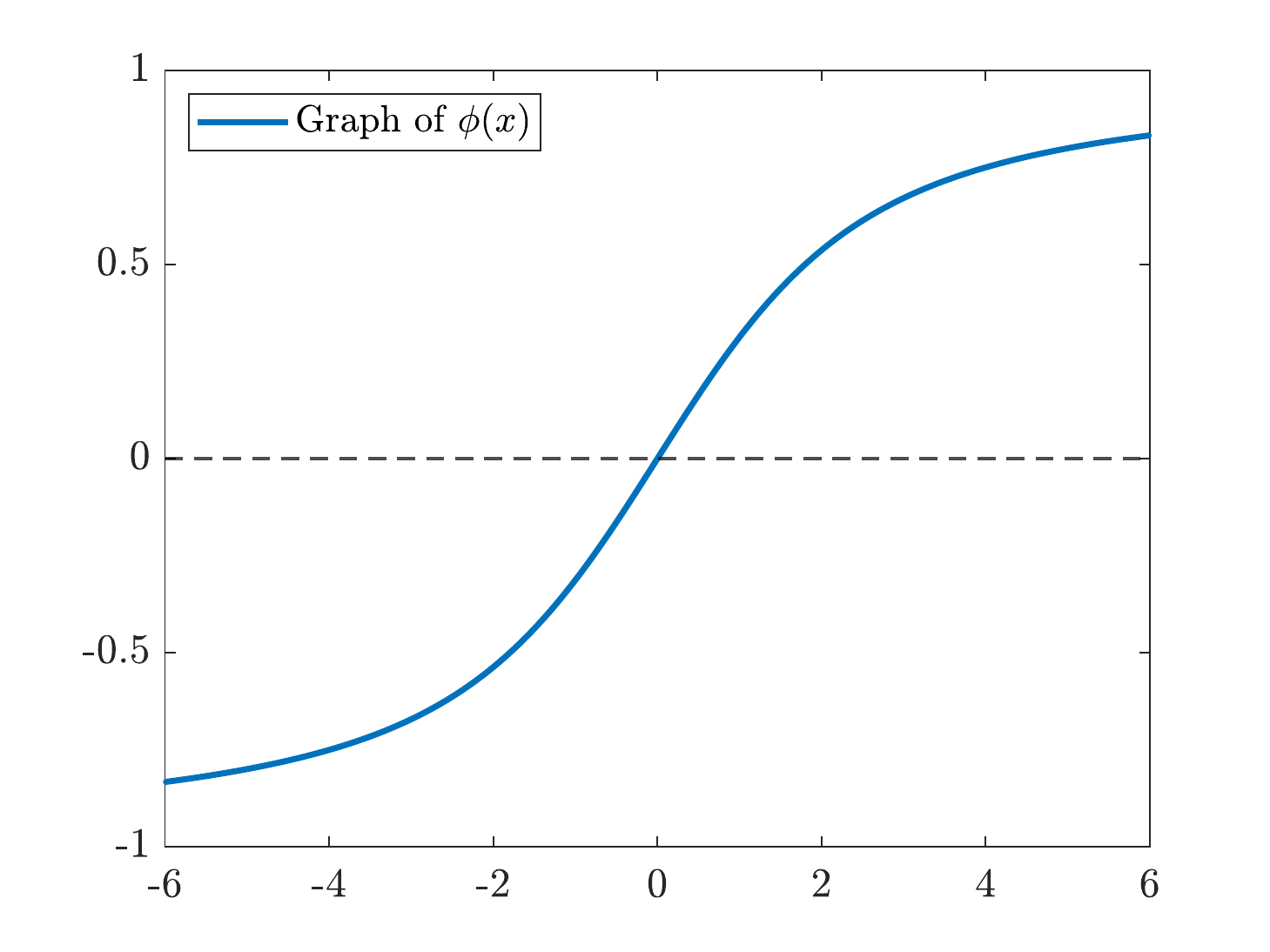}
\caption{Graph of $\phi(x)=\frac{\cosh x}{\sinh x}-\frac{1}{x}$.}\label{fig:phi_graph}
\end{figure}

We need the following lemma to prove Theorem \ref{thm:consistent}.

\begin{lem}\label{lem:lambda_bound}
Let $a<y<b$. Then the function
\begin{equation*}
J(\lambda)\coloneqq \begin{cases}
y\lambda-\log\left(\frac{\e^{\lambda b}-\e^{\lambda a}}{\lambda}\right) & (\lambda\neq 0)\\
-\log(b-a) & (\lambda=0)
\end{cases}
\end{equation*}
is strictly concave in $\lambda\in \R$ and achieves a unique maximum $\lambda^*$. Furthermore, letting $c=(a+b)/2$ and $d=b-a>0$, we have
\begin{equation}
\lambda^*=\frac{2}{d}\phi^{-1}\left(\frac{2(y-c)}{d}\right). \label{eq:lambda_bound}
\end{equation}
\end{lem}

\begin{proof}
The strict concavity of $J$ and the existence and uniqueness of $\lambda^*$ follow from Proposition \ref{prop:ME_density}. If $y=c$, then by \eqref{eq:lambdak_sign} we have $\lambda^*=0$ and \eqref{eq:lambda_bound} is trivial. Suppose $y\neq c$. Then by \eqref{eq:lambdak_sign} we have $\lambda^*\gtrless 0$ according as $y\gtrless c$.

When $\lambda\neq 0$, by definition
\begin{equation*}
J(\lambda)=(y-c)\lambda-\log\left(\frac{\e^{\lambda d/2}-\e^{-\lambda d/2}}{\lambda}\right).
\end{equation*}
Since $J$ is differentiable and $\lambda^*$ maximizes $J$, taking the derivative we obtain
\begin{equation*}
0=J'(\lambda^*)=y-c+\frac{1}{\lambda^*}-\frac{d}2\frac{\cosh(\lambda^*d/2)}{\sinh(\lambda^*d/2)}\iff \frac{\cosh x}{\sinh x}-\frac{1}{x}=\frac{2(y-c)}{d},
\end{equation*}
where $x=\lambda^*d/2$. Noting that $y\in (a,b)$, $c=(a+b)/2$, and $d=b-a$, we have $2(y-c)/d\in (-1,1)$. Therefore \eqref{eq:lambda_bound} follows from the definition of $\phi$ and Lemma \ref{lem:phi}.
\end{proof}

%%%%%%%%%%%%%%%%%%%%%%%%%%%%%%%%%%%%%%%%%%%%%%%%%%%%%%%%%%%%%
\setcounter{step}0
\begin{proof}[Proof of Theorem \ref{thm:consistent}]
We divide our proof into four steps.

\begin{step}
Reduction to the case when $C$ is an interval.
\end{step}

Since $f$ is Lipschitz, it is continuous, so the domain $D=\set{y\in \R:f(y)>0}$ is open. Let $C\subset D$ be compact. For each $y\in C$ we can take an open interval $(a_y,b_y)$ with $y\in (a_y,b_y)$ and $[a_y,b_y]\subset D$. Since $C$ is compact and $C\subset \bigcup_{y\in C}(a_y,b_y)$, we can take a finite subcover $C\subset \bigcup_{j=1}^J (a_{y_j},b_{y_j})$. Therefore $C\subset \bigcup_{j=1}^J [a_{y_j},b_{y_j}]\subset D$, so the uniform consistency on $C$ follows if we show the uniform consistency on each compact interval $[a,b]\subset D$.

Below, to simplify the argument, fix $C=[a,b]$ and $C'=[a',b']$ such that $C\subset C'\subset D$ and $a'<a<b<b'$. By relabeling the thresholds $\set{t_k}_{k=1}^K$ if necessary, without loss of generality we may assume $a'<t_K<a<b<t_1<b'$ by Condition \ref{cond:KN}.

\begin{step}\label{step:as}
Let $r_3\coloneqq \frac{1-4r_1}{2}>0$. As $n\to\infty$, we have
\begin{subequations}
\begin{align}
\max_{2\le k\le K}\abs{\frac{\hat{q}_k-q_k}{(t_{k-1}-t_k)^2}} &=o_\mathrm{a.s.}(1)=O_p(n^{-r_3}), \label{eq:estim_q}\\
\max_{2\le k\le K}\abs{\frac{\hat{y}_k-y_k}{t_{k-1}-t_k}} &=o_\mathrm{a.s.}(1)=O_p(n^{-r_3}). \label{eq:estim_y}
\end{align}
\end{subequations}
\end{step}

Let $(\Omega,\mathcal{F},\Pr)$ be the probability space over which the random variables $\set{Y_i}_{i=1}^\infty$ are defined. For any $\eta>0$, define the event
\begin{equation*}
E_n\coloneqq \set{\omega\in \Omega: \max_{2\le k\le K}\abs{\frac{\hat{q}_k-q_k}{(t_{k-1}-t_k)^2}}>\eta}.
\end{equation*}
Let $\id_{I_k}$ denote the indicator function of $I_k$. Using $\sum_{i=1}^n\id_{I_k}(Y_i)=n_k-n_{k-1}=n\hat{q}_k$, we obtain
\begin{align*}
\Pr(E_n) &\le \sum_{k=2}^K\Pr\left(\abs{ \frac{\hat{q}_k-q_k}{(t_{k-1}-t_k)^2}} >\eta \right)=\sum_{k=2}^K\Pr\left( \abs{ \frac{1}{n}\sum_{i=1}^n \frac{\id_{I_k}(Y_{i})-q_k}{(t_{k-1}-t_k)^2}} >\eta \right) 
\\
&\le \sum_{k=2}^K\Pr\left(\abs{\sum_{i=1}^n (\id_{I_k}(Y_{i}) -q_k)} >\eta c_1^2 n^{1-2r_1}\right),
\end{align*}
where the last line follows from \eqref{eq:tdiff_order}. Applying \cite{Hoeffding1963}'s inequality to the last term, we obtain
\begin{equation}
\Pr(E_n)\le \sum_{k=2}^K 2\exp(-2c_1^4\eta^2n^{2-4r_1}/n)=2(K-1)\exp(-2c_1^4\eta^2n^{1-4r_1}). \label{eq:PEn}
\end{equation}

Since $1-4r_1>0$ by Assumption \ref{asmp:primitive}\ref{cond:KN}, it follows that $\sum_{n=1}^\infty \Pr(E_n)<\infty$. By the Borel-Cantelli lemma, we have
\begin{equation*}
\Pr(E_n~\text{infinitely often})=\Pr\left(\bigcap_{m=1}^\infty \bigcup_{n=m}^\infty E_n\right)=0.
\end{equation*}
Since $\eta>0$ is arbitrary, we obtain $\max_{2\le k\le K}\abs{\frac{\hat{q}_k-q_k}{(t_{k-1}-t_k)^2}}\to 0$ almost surely as $n\to\infty$.

To show that this term is $O_p(n^{-r_3})$ for $r_3=\frac{1-4r_1}{2}$, for any $\epsilon>0$, choose $\delta>0$ such that $2(K-1)\exp(-2c_1^4\delta^2)<\epsilon$. Then the definition of $E_n$ and \eqref{eq:PEn} imply
\begin{equation*}
\Pr\left(n^\frac{1-4r_1}{2}\max_{2\le k\le K}\abs{\frac{\hat{q}_k-q_k}{(t_{k-1}-t_k)^2}}>\delta\right)<\epsilon
\end{equation*}
for all $n$, so by the definition $r_3$ and that of the order in probability, we obtain \eqref{eq:estim_q}.

We next estimate the term $\max_{2\le k\le K}\abs{\frac{\hat{y}_k-y_k}{t_{k-1}-t_k}}$. Using
\begin{equation*}
\sum_{i=1}^n Y_i\id_{I_k}(Y_i)=S_{n_k}-S_{n_{k-1}}=(n_k-n_{k-1})\hat{y}_k=n\hat{q}_k\hat{y}_k
\end{equation*}
and noting that $Y_i\in C'$ (compact set) whenever $\id_{I_k}(Y_i)>0$, a similar calculation yields
\begin{equation}
\max_{2\le k\le K}\abs{\frac{\hat{q}_k\hat{y}_k-q_ky_k}{(t_{k-1}-t_k)^2}} =O_p(n^{-r_3})=o_\mathrm{a.s.}(1). \label{eq:estim_qy}
\end{equation}

Using the triangle inequality and $\hat{y}_k\in I_k\subset [t_K,t_1]\subset [a',b']$, we obtain
\begin{equation}
q_k\abs{\hat{y}_k-y_k} \le \abs{\hat{q}_k\hat{y}_k-q_ky_k} + \hat{y}_k\abs{\hat{q}_k-q_k} \le \abs{\hat{q}_k\hat{y}_k-q_ky_k} + b'\abs{\hat{q}_k-q_k}. \label{eq:qabs_ydiff}
\end{equation}
Assumption \ref{asmp:primitive}\ref{cond:f} implies that we can take $0<m\le M<\infty$ such that $f:C'\to [m,M]$. 
By the mean value theorem for integrals, for each $k \in \set{2,\dots,K}$, there exists some $u_k \in [t_k,t_{k-1})$ such that
\begin{equation}
\frac{q_k}{t_{k-1}-t_k}=\frac{\int_{t_{k-1}}^{t_k}f(y) \diff y}{t_{k-1}-t_k}=f(u_k) \in [m,M]. \label{eq:q_tdiff_bound}
\end{equation}
Therefore dividing \eqref{eq:qabs_ydiff} by $(t_{k-1}-t_k)^2$ and using \eqref{eq:q_tdiff_bound}, we obtain
\begin{equation*}
m\abs{\frac{\hat{y}_k-y_k}{t_{k-1}-t_k}}\le \abs{ \frac{\hat{q}_k\hat{y}_k-q_ky_k}{(t_{k-1}-t_k)^2}} + b'\abs{\frac{\hat{q}_k-q_k}{(t_{k-1}-t_k)^2}}
\end{equation*}
for each $k \in \set{2,\dots,K}$.
Taking the maximum over $k$, dividing by $m>0$, and letting $n\to \infty$, we obtain \eqref{eq:estim_y}.

\begin{step}\label{step:ytbar}
Let $\bar{t}_k=(t_k+t_{k-1})/2$ be the midpoint of $I_k$, $m=\inf_{y\in C'}f(y)>0$, and $L>0$ be the Lipschitz constant for $f$. 
Then
\begin{equation}
\max_{2 \le k \le K}\abs{\frac{y_k-\bar{t}_k}{(t_{k-1}-t_k)^2}}\le \frac{L}{6m}. \label{eq:ytbar}
\end{equation}
\end{step}

Since by Assumption \ref{asmp:primitive}\ref{cond:f} $f$ is Lipschitz continuous with constant $L>0$, it is absolutely continuous. Theorem 3.35 of \cite{Folland1999} implies that $f$ is almost everywhere differentiable with $\abs{f'}\le L$ and we can apply the fundamental theorem of calculus (in particular, integration by parts). Therefore using $q_k=\int_{I_k}f(y)\diff y$ and $q_ky_k=\int_{I_k}yf(y)\diff y$,
we obtain
\begin{align*}
q_k(y_k-\bar{t}_k)&=\int_{I_k}(y-\bar{t}_k)f(y)\diff y\\
&=\left[\frac{1}{2}(y-\bar{t}_k)^2f(y)\right]_{t_k}^{t_{k-1}}-\int_{I_k}\frac{1}{2}(y-\bar{t}_k)^2f'(y)\diff y\\
&=\frac{1}{8}(t_{k-1}-t_k)^2(f(t_{k-1})-f(t_k))-\int_{I_k}\frac{1}{2}(y-\bar{t}_k)^2f'(y)\diff y.
\end{align*}
Therefore, for every $k \in \set{2,\dots,K}$, we have
\begin{align*}
q_k\abs{y_k-\bar{t}_k}&\le \frac{1}{8}(t_{k-1}-t_k)^2\abs{f(t_{k-1})-f(t_k)}+\int_{I_k}\frac{1}{2}(y-\bar{t}_k)^2L\diff y\\
&\le \frac{1}{8}(t_{k-1}-t_k)^3L+\frac{1}{24}(t_{k-1}-t_k)^3L\\
&=\frac{L}{6}(t_{k-1}-t_k)^3.
\end{align*}
Dividing both sides by $q_k(t_{k-1}-t_k)^2$ and using \eqref{eq:q_tdiff_bound}, we obtain
\begin{equation*}
\abs{\frac{y_k-\bar{t}_k}{(t_{k-1}-t_k)^2}}\le \frac{L}{6m}
\end{equation*}
uniformly over $k \in \set{2,\dots,K}$, which shows \eqref{eq:ytbar}.

\begin{step}\label{step:infeasible}
Strong consistency and the convergence rate of the estimator $\hat{f}(y)$ for $f(y)$ uniformly over $y \in C$.
\end{step}

Letting $I_k=[t_k,t_{k-1})$, by assumption we have $C\subset \bigcup_{k=2}^KI_k\subset C'$. Therefore for any $y \in C$, we can take $k$ such that $y\in I_k$. Using the moment conditions and the mean value theorem for the integrals, we can take $u(y),v(y)\in I_k$ such that
\begin{subequations}\label{eq:fuv}
\begin{align}
\hat{q}_k&=\int_{I_k}\hat{f}(y)\diff y=(t_{k-1}-t_k)\hat{f}(u(y)),\\
q_k&=\int_{I_k}f(y)\diff y=(t_{k-1}-t_k)f(v(y)).
\end{align}
\end{subequations}
Therefore,
\begin{align}
\abs{\hat{f}(y)-f(y)}&\le \abs{\hat{f}(y)-\hat{f}(u(y))}+\abs{\hat{f}(u(y))-f(v(y))}+\abs{f(v(y))-f(y)}\notag \\
&\eqqcolon A_1(y)+A_2(y)+A_3(y). \label{eq:A123}
\end{align}
Using Assumption \ref{asmp:primitive}\ref{cond:f}, $y,v(y)\in I_k$, and Assumption \ref{asmp:primitive}\ref{cond:KN}, we can uniformly bound $A_3$ as
\begin{align}
\sup_{y \in C} A_3(y)&=\sup_{y \in C} \abs{f(v(y))-f(y)}\le \sup_{y \in C} L\abs{v(y)-y} \notag
\\
&\le L\max_{2 \le k \le K} (t_{k-1}-t_k)
\le Lc_2n^{-r_2}.\label{eq:A3}
\end{align}

Using \eqref{eq:fuv}, Step \ref{step:as}, and Assumption \ref{asmp:primitive}\ref{cond:KN}, we can uniformly bound $A_2$ as
\begin{align}
\sup_{y \in C} A_2(y)&= \max_{2 \le k \le K} \abs{\frac{\hat{q}_k-q_k}{t_{k-1}-t_k}}\notag \\
&\le c_2n^{-r_2} \max_{2\le k\le K}\abs{\frac{\hat{q}_k-q_k}{(t_{k-1}-t_k)^2}}=c_2n^{-r_2}O_p(n^{-r_3}). \label{eq:A2}
\end{align}

Noting that $\hat{f}$ is exponential on $I_k$ and hence monotonic, it follows from the functional form in \eqref{eq:fstar} and the mean value theorem that
\begin{align}
\sup_{y \in C} A_1(y)=&\sup_{y \in C} \abs{\hat{f}(y)-\hat{f}(u(y))}\le \max_{2 \le k \le K}\abs{\hat{f}(t_{k-1})-\hat{f}(t_k)} \notag \\
\le& \max_{2 \le k \le K} \hat{q}_k\abs{\lambda_k}=\max_{2 \le k \le K}\frac{\hat{q}_k}{t_{k-1}-t_k} \max_{2 \le k \le K} \abs{\lambda_k}(t_{k-1}-t_k).\label{eq:A1}
\end{align}
By Step \ref{step:as} and \eqref{eq:q_tdiff_bound}, we can bound the first term in the right-hand side of \eqref{eq:A1} as
\begin{equation}
\max_{2 \le k \le K}\frac{\hat{q}_k}{t_{k-1}-t_k}\le \max_{2 \le k \le K}\frac{q_k}{t_{k-1}-t_k}+\max_{2 \le k \le K}\abs{\frac{\hat{q}_k-q_k}{t_{k-1}-t_k}}\le M+O_p(n^{-r_3}). \label{eq:A1ub1}
\end{equation}
To bound the second term, note from Steps \ref{step:as} and \ref{step:ytbar} that
\begin{align*}
\max_{2 \le k \le K}\abs{\frac{\hat{y}_k-\bar{t}_k}{t_{k-1}-t_k}}&\le \max_{2 \le k \le K}\abs{\frac{\hat{y}_k-y_k}{t_{k-1}-t_k}}+\max_{2 \le k \le K}\abs{\frac{y_k-\bar{t}_k}{(t_{k-1}-t_k)^2}}(t_{k-1}-t_k)\\
&\le O_p(n^{-r_3})+\frac{L}{6m}c_2n^{-r_2}.
\end{align*}
Therefore it follows from Lemma \ref{lem:lambda_bound} that
\begin{align}
\max_{2 \le k \le K}\abs{\lambda_k}(t_{k-1}-t_k) &\le  2\phi^{-1}\left(\max_{2 \le k \le K} 2\abs{\frac{\hat{y}_k-\bar{t}_k}{t_{k-1}-t_k}}\right)\notag \\
&\le 2\phi^{-1}\left(\frac{L}{3m}c_2n^{-r_2}+O_p(n^{-r_3})\right).\label{eq:A1ub2}
\end{align}
Combining \eqref{eq:A1}, \eqref{eq:A1ub1} and \eqref{eq:A1ub2}, we obtain
\begin{equation}
\sup_{y \in C} A_1(y)\le 2(M+O_p(n^{-r_3}))\phi^{-1}\left(\frac{L}{3m}c_2n^{-r_2}+O_p(n^{-r_3})\right).\label{eq:A1ub3}
\end{equation}
The uniform bound \eqref{eq:rate} follows from \eqref{eq:A123}, \eqref{eq:A3}, \eqref{eq:A2}, \eqref{eq:A1ub3}, and Lemma \ref{lem:phi}. Noting that the $O_p(n^{-r_3})$ term is also $o_\mathrm{a.s.}(1)$ in this case, we also obtain the uniform strong consistency \eqref{eq:consistent}.
\end{proof}

%%%%%%%%%%%%%%%%%%%%%%%%%%%%%

\begin{proof}[Proof of Corollary \ref{cor:quantile}]
Since $Q_\tau \in \interior C$, there exists a $\delta_1$-ball $B_{\delta_1}(Q_\tau)$ around $Q_\tau$ such that $B_{\delta_1}(Q_\tau) \subset C$.
Also, since $Q_\tau \in \interior C \subset D$, $f(Q_\tau)>0$ by Assumption \ref{asmp:primitive}\ref{cond:f}.
Furthermore, $f$ is continuous by Assumption \ref{asmp:primitive}\ref{cond:f}. 
Therefore, $f$ is bounded away from zero in a neighborhood of $Q_\tau$.
That is, there exists a $\delta_2$-ball $B_{\delta_2}(Q_\tau)$ and $\ubar{f}>0$ such that $f(y) \ge \ubar{f}(y)$ for all $y \in B_{\delta_2}(Q_\tau)$.
Let $\delta = \frac{1}{2} \left(\delta_1 \wedge \delta_2\right)$.

We first show that
\begin{equation}\label{eq:quantile1}
F(Q_\tau - \delta) + \delta \cdot \ubar{f} \le  \hat F(\hat{Q}_\tau) \le F(Q_\tau + \delta) - \delta \ubar{f}.
\end{equation}
By way of contradiction, suppose that
$\hat F(\hat{Q}_\tau) > F(Q_\tau + \delta) - \delta \ubar{f}$.
Then, since $F$ and $\hat F$ are continuous by construction,
\begin{equation*}
F(Q_\tau) = \tau = \hat F(\hat{Q}_\tau) > F(Q_\tau + \delta) - \delta \ubar{f}
\end{equation*}
and hence
$\delta \ubar{f} \ge F(Q_\tau+\delta) - F(Q_\tau) > \delta \ubar{f}$,
which is a contradiction.
Therefore, we must have
$\hat F(\hat{Q}_\tau) \le F(Q_\tau + \delta) - \delta \ubar{f}$.
Sinimarly,
$F(Q_\tau - \delta) + \delta \ubar{f} \le \hat F(\hat{Q}_\tau)$ must hold too.
Therefore, \eqref{eq:quantile1} holds.

Since $F$ is non-decreasing by construction, and it is also strictly increasing on $B_{\delta}(Q_\tau) \subset B_{\delta_2}(Q_\tau)$, \eqref{eq:quantile1} implies  
\begin{equation}\label{eq:quantile2}
\hat{Q}_\tau \in B_\delta(Q_\tau) \subset C \quad 
\text{provided that} \quad 
\sup_{y \in C} \abs{\hat F(y) - F(y)} < \delta \ubar{f}.
\end{equation}
Also, noting that $\hat F(\hat{Q}_\tau) = \tau = F(Q_\tau)$, we obtain
\begin{equation}\label{eq:quantile3}
\hat F(\hat{Q}_\tau) - F(\hat{Q}_\tau) +  F(\hat{Q}_\tau) -  F(Q_\tau) = 0.
\end{equation}

With all these pieces put together, we get
\begin{align*}
& P\left(\limsup_{n \to \infty}\left\{ \omega \in \Omega : \abs{\hat{Q}_\tau - Q_\tau} = 0 \right\} \right)
\\
\stackrel{\text{(i)}}{=}& P\left(\limsup_{n \to \infty}\left\{ \omega \in \Omega : \abs{\int_{Q_\tau}^{\hat{Q}_\tau} f(y) \diff y} = 0 \right\} \right)
\\
=& P\left(\limsup_{n \to \infty}\left\{ \omega \in \Omega : \abs{F(\hat{Q}_\tau) - F(Q_\tau)} = 0 \right\} \right)
\\
\stackrel{\text{(ii)}}{=}& P\left(\limsup_{n \to \infty}\left\{ \omega \in \Omega : \abs{\hat F(\hat{Q}_\tau) - F(\hat{Q}_\tau) } = 0 \right\} \right)
\\
\ge& P\left(\limsup_{n \to \infty}\left\{ \omega \in \Omega :  \sup_{y \in C}\abs{\hat F(y) - F(y)} = 0~\&~\hat{Q}_\tau \in C \right\}\right)
\\
\stackrel{\text{(iii)}}{\ge}& P\left(\limsup_{n \to \infty}\left\{ \omega \in \Omega :  \sup_{y \in C}\abs{\hat F(y) - F(y)} = 0~\&~\sup_{y \in C}\abs{\hat F(y) - F(y)} < \delta \ubar{f} \right\}\right)
\\
=& P\left(\limsup_{n \to \infty}\left\{ \omega \in \Omega :  \sup_{y \in C}\abs{\hat F(y) - F(y)} = 0 \right\}\right)
\\
\stackrel{\text{(iv)}}{=}& 1.
\end{align*}
Here equality (i) follows from $f(y) > \ubar{f}$ for all $y \in B_\delta(Q_\tau)$,
equality (ii) follows from \eqref{eq:quantile3},
inequality (iii) follows from \eqref{eq:quantile2}, and 
equality (iv) follows from Corollary \ref{cor:cdf}.
This completes a proof of the corollary.
\end{proof}

\begin{proof}[Proof of Theorem \ref{thm:ME_threshold}]
\setcounter{step}{0}
Let $\cl T=[y_2,y_1]\times [y_3,y_2] \times\dots\times [y_K,y_{K-1}]$ be the closure of $T$ in \eqref{eq:T}, which is a nonempty compact interval. Let $\partial T\coloneqq \cl T\backslash T$ be the boundary of $T$. Extend $J^*$ to $\cl T$ by defining $J^*=\infty$ on $\partial T$.

We divide our proof into three steps.

\begin{step}
$J^*:\cl T\to \R\cup\set{\infty}$ is lower semicontinuous.
\end{step}

Motivated by \eqref{eq:Jk} and \eqref{eq:Jstar}, define $J:\R^K\times \cl T\to (-\infty,\infty]$ by
\begin{equation*}
J(\lambda;t)=\begin{cases}
\sum_{k=1}^Kq_k\left(y_k\lambda_k-\log\left(\frac{\e^{\lambda_k t_{k-1}}-\e^{\lambda_k t_k}}{\lambda_k}\right)+\log q_k\right), & (t\in T)\\
\infty, & (t\in \partial T)
\end{cases}
\end{equation*}
where the case $\lambda_k=0$ needs to be separately defined in the obvious way using \eqref{eq:Jk}. Then we can easily verify that $J(\lambda;\cdot): \cl T\to (-\infty,\infty]$ is continuous. Therefore $J^*(t)=\sup_{\lambda}J(\lambda;t)$ is lower semicontinuous.

\begin{step}
$J^*:T\to \R$ is strictly convex.
\end{step}

Since $J^*(t)=\sup_{\lambda}J(\lambda;t)$ and the supremum of convex functions is convex, it suffices to show that $J(\lambda;t)$ is strictly convex in $t\in T$. To this end it suffices to show that its Hessian $\nabla^2_t J$ is positive definite almost everywhere on $T$. A tedious but straightforward calculation yields
\begin{subequations}\label{eq:J_derivative}
\begin{align}
\frac{\partial J}{\partial t_k}&=\frac{q_k\lambda_k}{\e^{\lambda_k(t_{k-1}-t_k)}-1}-\frac{q_{k+1}\lambda_{k+1}}{1-\e^{\lambda_{k+1}(t_{k+1}-t_k)}}, \label{eq:DJtk}\\
\frac{\partial^2 J}{\partial t_k^2}&=\frac{q_k\lambda_k^2 \e^{\lambda_k(t_{k-1}-t_k)}}{(\e^{\lambda_k(t_{k-1}-t_k)}-1)^2}+\frac{q_{k+1}\lambda_{k+1}^2 \e^{\lambda_{k+1}(t_{k+1}-t_k)}}{(1-\e^{\lambda_{k+1}(t_{k+1}-t_k)})^2},\\
\frac{\partial^2 J}{\partial t_k \partial t_{k+1}}&=-\frac{q_{k+1}\lambda_{k+1}^2 \e^{\lambda_{k+1}(t_{k+1}-t_k)}}{(1-\e^{\lambda_{k+1}(t_{k+1}-t_k)})^2},\\
\frac{\partial^2 J}{\partial t_k \partial t_{k+\ell}}&=0 \quad \text{for $\ell>1$},
\end{align}
\end{subequations}
which are valid for $k=1$ by setting $\e^{\lambda_1t_0}=0$. Collecting the partial derivatives \eqref{eq:J_derivative} into a matrix, we obtain the Hessian
\begin{equation}
\nabla^2_t J=H_{K-1}\coloneqq \begin{bmatrix}
c_1+c_2 & -c_2 & 0 & \cdots & 0\\
-c_2 & c_2+c_3 & -c_3 & \ddots & \vdots\\
0 & \ddots & \ddots & \ddots & 0\\
\vdots & \ddots & -c_{K-2} & c_{K-2} + c_{K-1} & -c_{K-1} \\
0 & \cdots & 0 & -c_{K-1} & c_{K-1}+c_K
\end{bmatrix},\label{eq:J_hessian}
\end{equation}
where
\begin{equation*}
c_k\coloneqq \frac{q_k\lambda_k^2 \e^{\lambda_k(t_{k-1}-t_k)}}{(\e^{\lambda_k(t_{k-1}-t_k)}-1)^2}\ge 0,
\end{equation*}
with strict inequality if $\lambda_k\neq 0$. Since by \eqref{eq:lambdak_sign} we have $\lambda_k=0$ if and only if $y_k=\frac{t_k+t_{k-1}}2$, we have $c_k>0$ for all $k$ almost everywhere on $T$. Below, consider such $t$.

As is well known, a real symmetric matrix is positive definite if and only if all leading principal minors are positive. In general, for a tridiagonal matrix
\begin{equation*}
A_n\coloneqq \begin{bmatrix}
d_1 & a_1 & 0 & \cdots & 0\\
b_2 & d_2 & a_2 & \ddots & \vdots\\
0 & \ddots & \ddots & \ddots & 0\\
\vdots & \ddots & b_{n-1} & d_{n-1} & a_{n-1} \\
0 & \cdots & 0 & b_n & d_n
\end{bmatrix},
\end{equation*}
its determinant $D_n\coloneqq \det A_n$ satisfies the three-term recurrence relation 
\begin{equation}
D_n=d_nD_{n-1}-b_na_{n-1}D_{n-2}, \label{eq:TTRR}
\end{equation}
where $D_0=1$ and $D_1=d_1$ \citep{El_Mikkawy_2004}. Applying \eqref{eq:TTRR} to $H_{K-1}$ in \eqref{eq:J_hessian} with $n=K-1$, since $d_n=c_n+c_{n+1}$ and $b_n=a_{n-1}=-c_n$, we obtain
\begin{equation}
D_n=(c_n+c_{n+1})D_{n-1}-c_n^2D_{n-2}. \label{eq:TTRR_H}
\end{equation}
Clearly $D_1=d_1=c_1+c_2>0$. Let us show by induction that
\begin{equation}
D_n=c_{n+1}D_{n-1}+\prod_{k=1}^n c_k,\label{eq:fn}
\end{equation}
which implies $D_n>0$ for all $n$. If $n=1$, then $D_1=c_1+c_2=c_2D_0+c_1$ because $D_0=1$, so \eqref{eq:fn} holds. Suppose \eqref{eq:fn} holds for some $n$. Then for $n+1$, using \eqref{eq:TTRR_H} and the induction hypothesis, we obtain
\begin{align*}
D_{n+1}&=(c_{n+1}+c_{n+2})D_n-c_{n+1}^2D_{n-1}\\
&=c_{n+2}D_n+c_{n+1}\left(c_{n+1}D_{n-1}+\prod_{k=1}^n c_k\right)-c_{n+1}^2D_{n-1}\\
&=c_{n+2}D_n+\prod_{k=1}^{n+1}c_k,
\end{align*}
so \eqref{eq:fn} holds for $n+1$. This completes the proof of the strict convexity of $J(\lambda;t)$ in $t$, and hence $J^*(t)=\max_{\lambda} J(\lambda;t)$ is strictly convex in $t\in T$.

\begin{step}
There exists a unique $t^*\in T$ that minimizes $J^*$. For this $t^*$, the maximum entropy density $f^*$ in \eqref{eq:fstar} is continuous on $[t_K,\infty)$.
\end{step}

Since $\cl T$ is nonempty compact, $J^*:\cl T\to (-\infty,\infty]$ is lower semicontinuous, and $J^*$ is finite-valued on $T$, it achieves a minimum at some $t^*\in \cl T$. Since $J^*=\infty$ on $\partial T$, it must be $t^*\in T$. Since $J^*$ is strictly convex on $T$, the minimizer $t^*$ is unique. 

Let us next show the continuity of $f^*$. Since $f^*$ is piecewise exponential (hence continuous), it suffices to show the continuity at the thresholds $t^*_k$. Since $J(\lambda;\cdot):T\to \R$ is differentiable, it follows from the first-order condition, the envelope theorem, and \eqref{eq:DJtk} that
\begin{align}
& 0=\frac{\partial J^*}{\partial t_k}(t^*)=\frac{\partial J}{\partial t_k}(t^*)=\frac{q_k\lambda_k}{\e^{\lambda_k(t^*_{k-1}-t^*_k)}-1}-\frac{q_{k+1}\lambda_{k+1}}{1-\e^{\lambda_{k+1}(t^*_{k+1}-t^*_k)}} \notag \\
\iff & \frac{q_k\lambda_k\e^{\lambda_kt^*_k}}{\e^{\lambda_kt^*_{k-1}}-\e^{\lambda_kt^*_k}}=\frac{q_{k+1}\lambda_{k+1}\e^{\lambda_{k+1}t^*_k}}{\e^{\lambda_{k+1}t^*_k}-\e^{\lambda_{k+1}t^*_{k+1}}}.\label{eq:foc_tk}
\end{align}
Comparing \eqref{eq:fstar} and \eqref{eq:foc_tk}, we obtain $f^*(t^*_k-)=f^*(t^*_k+)$, so $f^*$ is continuous.
\end{proof}

%%%%%%%%%%%%%%%%%%%%%%
\section{Numerical issues}\label{sec:numerical}
%%%%%%%%%%%%%%%%%%%%%%

Although conceptually straightforward, numerically maximizing $J_k$ in \eqref{eq:Jk} over $\lambda$ can be unstable because the values of $\set{(t_k,y_k,t_{k-1})}_{k=1}^K$ change by many orders of magnitude in typical data, for instance $t_K=1$ and $t_1=10^7$ in Table \ref{t:IRS2019}. For this reason it is useful to choose some scaling factor $s>0$ and consider $(st_k,sy_k,st_{k-1})$. Define
\begin{equation*}
J_k(\lambda;s)\coloneqq sy_k\lambda-\log\left(\frac{\e^{\lambda st_{k-1}}-\e^{\lambda st_k}}{\lambda}\right)=J_k(\lambda s)-\log s.
\end{equation*}
Maximizing both sides over $\lambda\in \R$ and using the definition of $\lambda_k$ in Proposition \ref{prop:ME_density}, it follows that
\begin{equation*}
\lambda_k=s \argmax_{\lambda\in \R}J_k(\lambda;s) \quad \text{and} \quad J_k(\lambda_k)=\max_{\lambda\in \R}J_k(\lambda;s)+\log s.
\end{equation*}
Therefore to numerically compute $\lambda_k$ and $J_k(\lambda_k)$, we may maximize $J_k(\lambda,s)$ for some scaling factor $s>0$ (say $s=1/t_{k-1}$ so that $st_k<sy_k<st_{k-1}=1$), multiply its maximizer by $s$, and add $\log s$ to its maximum value.

%%%%%%%%%%%%%%%%%%%%%%%%%%%%%%%%%%%%%%%%%%%%
\section{Lorenz curve for double Pareto distribution}\label{sec:dPLorenz}
%%%%%%%%%%%%%%%%%%%%%%%%%%%%%%%%%%%%%%%%%%%%

Integrating the double Pareto density (with $M=1$), it is easy to show that the CDF of double Pareto is
\begin{equation*}
F(y)=\begin{cases}
\frac{\alpha}{\alpha+\beta}y^\beta, & (0\le y\le 1)\\
1-\frac{\beta}{\alpha+\beta}y^{-\alpha}. & (y\ge 1)
\end{cases}
\end{equation*}

Let us compute the Lorenz curve, which is implicitly defined by
\begin{equation*}
L(F(y))=\frac{\int_0^y tf(t)\diff t}{\mu},
\end{equation*}
where $\mu$ is the mean. If $y\le 1$, then
\begin{equation*}
\int_0^y tf(t)\diff t=\int_0^y \frac{\alpha\beta}{\alpha+\beta}t^\beta \diff t=\frac{\alpha\beta}{(\alpha+\beta)(\beta+1)}y^{\beta+1}.
\end{equation*}
If $y\ge 1$, then
\begin{align*}
\int_0^x tf(t)\diff t&=\int_0^1 tf(t)\diff t+\int_1^y tf(t)\diff t\\
&=\frac{\alpha\beta}{(\alpha+\beta)(\beta+1)}+\int_1^y \frac{\alpha\beta}{\alpha+\beta}t^{-\alpha} \diff t\\
&=\frac{\alpha\beta}{(\alpha+\beta)(\alpha-1)}\left(\frac{\alpha+\beta}{\beta+1}-y^{-\alpha+1}\right).
\end{align*}
Letting $y\to\infty$, the mean of double Pareto (assuming $\alpha>1$) is
\begin{equation*}
\mu=\frac{\alpha\beta}{(\alpha-1)(\beta+1)}.
\end{equation*}

Thus if $y\le 1$ (or equivalently $F=F(y)\le \frac{\alpha}{\alpha+\beta}$), then
\begin{align*}
L(F)&=\frac{\alpha-1}{\alpha+\beta}y^{\beta+1}=\frac{\alpha-1}{\alpha+\beta}\left(\frac{\alpha+\beta}{\alpha}F\right)^{1+1/\beta}\\
&=(\alpha-1)(\alpha+\beta)^{1/\beta}\alpha^{-1-1/\beta}F^{1+1/\beta}.
\end{align*}
If $y\ge 1$ (or equivalently $F=F(y)\ge \frac{\alpha}{\alpha+\beta}$), then
\begin{align*}
L(F)&=1-\frac{\beta+1}{\alpha+\beta}y^{-\alpha+1}=1-\frac{\beta+1}{\alpha+\beta}\left(\frac{\alpha+\beta}{\beta}(1-F)\right)^{1-1/\alpha}\\
&=1-(\beta+1)(\alpha+\beta)^{-1/\alpha}\beta^{-1+1/\alpha}(1-F)^{1-1/\alpha}.
\end{align*}
Putting all the pieces together, the Lorenz curve for a double Pareto distribution with $\alpha>1$ and $\beta>0$ is
\begin{equation*}
L(x)=\begin{cases}
(\alpha-1)(\alpha+\beta)^{1/\beta}\alpha^{-1-1/\beta}x^{1+1/\beta}, & (0\le x\le \frac{\alpha}{\alpha+\beta})\\
1-(\beta+1)(\alpha+\beta)^{-1/\alpha}\beta^{-1+1/\alpha}(1-x)^{1-1/\alpha}. & (\frac{\alpha}{\alpha+\beta}\le x\le 1)
\end{cases}
\end{equation*}

%%%%%%%%%%%%%%%%%%%%%%%%%%%%%%%%%%%%%%%%%%%%
\section{\cite{Piketty2003}'s Pareto interpolation method}\label{sec:piketty}
%%%%%%%%%%%%%%%%%%%%%%

\cite{Piketty2003}'s Pareto interpolation method for constructing top income shares is widely used, for example in \cite{PikettySaez2003} and the World Inequality Database.\footnote{\url{https://wid.world/}} Because the detailed description of the method is relegated to \citet[Appendix B, Section 1.1, pp.~592--599]{Piketty2001book}, which is written in French, for completeness we describe the method here.

We use the same notation as in Section \ref{subsec:ME}. 
For simplicity assume that we are interested in the income distribution and the top income shares of the taxpayers, so the sample size is $n=n_K$. 
Suppose we observe the lower income threshold for income group $k$, which we denote by $t_k\coloneqq Y_{(n_k)}$. Let $p_k\coloneqq n_k/n$ be the top fractile corresponding to the $k$-th income threshold $t_k$ and
\begin{equation*}
s_k\coloneqq \frac{1}{n_k}\sum_{i=1}^{n_k}Y_{(i)}=\frac{S_{n_k}}{n_k}
\end{equation*}
be the average income of taxpayers with income above $t_k$.

Let $b_k=s_k/t_k$ be the ratio between the average income of individuals exceeding $t_k$ and the income threshold $t_k$. If income $Y$ is Pareto distributed (in the upper tail), then for large enough income level $y$, the CDF takes the form $F(y)=1-Ay^{-\alpha}$ for some $A>0$ and Pareto exponent $\alpha>1$. Therefore for any large enough income threshold $t$, we have
\begin{equation}
b(t)\coloneqq \frac{\E[Y \mid Y\ge t]}{t}=\frac{\int_t^\infty yF'(y)\diff y}{t(1-F(t))}=\frac{\alpha}{\alpha-1}. \label{eq:local_coeff}
\end{equation}
\cite{Piketty2001book} refers to $b_k=s_k/t_k$ as the (local) Pareto \emph{coefficient}. When the income distribution has a Pareto upper tail, the (local) Pareto \emph{exponent} can be recovered from \eqref{eq:local_coeff} as
\begin{equation}
\alpha_k=\frac{b_k}{b_k-1}=\frac{1}{1-t_k/s_k}. \label{eq:local_exp}
\end{equation}

Let $p\in (0,1]$ and suppose we would like to construct the top $p$ fractile income share. For example, $p=0.01$ corresponds to the top 1\% income share. \cite{Piketty2001book} proceeds as follows to construct the top $p$ fractile income share. First, let $p_k$ be the closest proportion to $p$ directly observed in the data, and $t_k$ and $\alpha_k$ be the corresponding income threshold and local Pareto exponent. Then one supposes that the income distribution is locally exactly Pareto, and therefore the CDF is
\begin{equation*}
F(y)=1-p_k(y/t_k)^{-\alpha_k}.
\end{equation*}
The income threshold corresponding to the top $p$ fractile can be computed as
\begin{equation*}
1-p=1-p_k(y/t_k)^{-\alpha_k}\iff y=t(p)\coloneqq t_k(p_k/p)^{1/\alpha_k}.
\end{equation*}
Noting that the sample size is $n$, the total income of taxpayers in the top $p$ fractile can be computed as
\begin{align}
S(p)&\coloneqq n\int_{t(p)}^\infty yF'(y)\diff y=n\int_{t(p)}^\infty \alpha_kp_k(y/t_k)^{-\alpha_k}\diff y \notag \\
&=n\frac{\alpha_k}{\alpha_k-1}p_kt_k^{\alpha_k}t(p)^{1-\alpha_k}=n\frac{\alpha_k}{\alpha_k-1}pt_k(p_k/p)^{1/\alpha_k} \notag \\
&=ns_kp_k^{1/\alpha_k}p^{1-1/\alpha_k}. \label{eq:Yq}
\end{align}
The top $p$ fractile income share can then be computed as $S(p)/S(1)$. If the researcher is interested in the top income share of all income earners (including those who do not file for taxes), one would adjust the sample size $n$ and the total income $S(1)$ from other sources (\eg, population statistics and national accounts); see for example \citet[Appendix A]{PikettySaez2001WP}.

%%%%%%%%%%%%%%%%%%%%%%%%%%%%%%%%%%%%%%%%%%%%%%%%%%%%%
\section{Additional simulations}\label{sec:additional}

In the main text, we used the double Pareto distribution to evaluate the bias and RMSE in the top income shares. Here, we repeat the analysis in Table \ref{t:simu} with data generated from lognormal, gamma, and Weibull distributions, which are described in Table \ref{t:model}. 
For the lognormal distribution, we set $\sigma = 1.5$. 
For the gamma and the Weibull distributions, we set $a=1$ and $k=1$, respectively, so that they are both identical to the exponential distribution with unit mean (for which it is possible to compute top income shares in closed-form). 
Other setups and the four estimation methods are the same as those described in Section \ref{sec:Simulation}.
Tables \ref{t:simu_lognormal} and \ref{t:simu_gamma} present the results.

%%%%%%%%%%%%%%%%%%%%%%
\begin{table}[!htb]
\centering
${}$
\caption{Simulation results of four methods for top shares with lognormal data.}\label{t:simu_lognormal}
\resizebox{\textwidth}{!}{
\begin{tabular}{lrrrrrrrrrrrr}
\toprule

\multicolumn{13}{l}{\bf Relative Bias}\\
$p_0$ & 0.001 & 0.01 & 0.05 & 0.1 & 0.2 & 0.3 & 0.4 & 0.5 & 0.6 & 0.7 & 0.8 & 0.9 \\
\midrule
$n$             & \multicolumn{12}{c}{\bf ME} \\
$10^4$ &-0.016 &-0.003 &-0.001	 &0.000	 &0.000	 &0.000	 &0.000	  &0.000	 &0.000	 &0.000	 &0.000	 &0.000\\
$10^5$ &-0.003 &0.000 &0.000	 &0.000	 &0.000	 &0.000	 &0.000	  &0.000	 &0.000	 &0.000	 &0.000	 &0.000\\
$10^6$ &-0.002 &0.00 &0.000	 &0.000	 &0.000	 &0.000	 &0.000	  &0.000	 &0.000	 &0.000	 &0.000	 &0.000\\
\midrule
$n$            & \multicolumn{12}{c}{\bf KP} \\
$10^4$ & 0.897   &	0.112	&-0.036	&-0.034	&-0.012	&0.001	&0.006	&0.006	&0.004	&0.002	&0.000	&-0.001\\
$10^5$ & 0.907   &	0.113	&-0.035	&-0.034	&-0.011	&0.001	&0.006	&0.006	&0.004	&0.002	&0.000	&-0.001\\
$10^6$ & 0.908   &	0.113	&-0.035	&-0.034	&-0.011	&0.001	&0.006	&0.006	&0.004	&0.002	&0.000	&-0.001\\
\midrule
$n$          & \multicolumn{12}{c}{\bf VA} \\
$10^4$ &	0.412	&0.044	&-0.016	&-0.013	&-0.004	&0.001	&0.003	&0.004	&0.003	&0.002	&0.001	&0.000\\
$10^5$ &	0.411	&0.044	&-0.015	&-0.012	&-0.003	&0.002	&0.003	&0.004	&0.003	&0.002	&0.001	&0.000\\
$10^6$ &	0.410	&0.044	&-0.015	&-0.012	&-0.003	&0.002	&0.003	&0.004	&0.003	&0.002	&0.001	&0.000\\
\midrule
$n$          & \multicolumn{12}{c}{\bf HGBRC} \\
$10^4$ &	0.296	&0.078	&0.018	&0.007	&0.002	&0.001	&0.000	&0.000	&0.000	&0.000	&0.000	&0.000\\
$10^5$ &	0.256	&0.084	&0.026	&0.013	&0.005	&0.003	&0.002	&0.001	&0.001	&0.000	&0.000	&0.000\\
$10^6$ &	0.235	&0.080	&0.026	&0.013	&0.005	&0.003	&0.002	&0.001	&0.001	&0.000	&0.000	&0.000\\

\toprule
\multicolumn{13}{l}{\bf Relative RMSE}\\
$p_0$ & 0.001 & 0.01 & 0.05 & 0.1 & 0.2 & 0.3 & 0.4 & 0.5 & 0.6 & 0.7 & 0.8 & 0.9 \\
\midrule
$n$             & \multicolumn{12}{c}{\bf ME} \\
$10^4$ &	0.342	&0.106	&0.037	&0.022	&0.011	&0.006	&0.004	&0.002	&0.001	&0.001	&0.000	&0.000\\
$10^5$ &	0.109	&0.033	&0.012	&0.007	&0.004	&0.002	&0.001	&0.001	&0.000	&0.000	&0.000	&0.000\\
$10^6$ &	0.036	&0.010	&0.004	&0.002	&0.001	&0.001	&0.000	&0.000	&0.000	&0.000	&0.000	&0.000\\
\midrule
$n$             & \multicolumn{12}{c}{\bf KP} \\
$10^4$ &	0.932	&0.129	&0.046	&0.040	&0.017	&0.008	&0.008	&0.007	&0.005	&0.002	&0.000	&0.001\\
$10^5$ &	0.911	&0.115	&0.037	&0.034	&0.012	&0.003	&0.006	&0.006	&0.004	&0.002	&0.000	&0.001\\
$10^6$ &	0.909	&0.113	&0.036	&0.034	&0.011	&0.002	&0.006	&0.006	&0.004	&0.002	&0.000	&0.001\\
\midrule
$n$        & \multicolumn{12}{c}{\bf VA} \\
$10^4$ &	0.423	&0.071	&0.038	&0.027	&0.013	&0.008	&0.005	&0.004	&0.003	&0.002	&0.001	&0.000\\
$10^5$ &	0.412	&0.047	&0.019	&0.014	&0.005	&0.003	&0.004	&0.004	&0.003	&0.002	&0.001	&0.000\\
$10^6$ &	0.411	&0.044	&0.016	&0.013	&0.003	&0.002	&0.003	&0.004	&0.003	&0.002	&0.001	&0.000\\
\midrule
$n$          & \multicolumn{12}{c}{\bf HGBRC} \\
$10^4$ &	0.587	&0.190	&0.067	&0.037	&0.018	&0.010	&0.006	&0.004	&0.002	&0.001	&0.001	&0.000\\
$10^5$ &	0.313	&0.101	&0.032	&0.016	&0.007	&0.004	&0.002	&0.001	&0.001	&0.001	&0.000	&0.000\\
$10^6$ &	0.254	&0.084	&0.026	&0.013	&0.006	&0.003	&0.002	&0.001	&0.001	&0.000	&0.000	&0.000\\

\bottomrule
\end{tabular}
}
\caption*{\footnotesize Note: See Table \ref{t:simu} for details.}
\end{table}

%%%%%%%%%%%%%%%%%%%%%%%%%%%
\begin{table}[!htb]
\centering
${}$
\caption{Simulation results of four methods for top shares with gamma data.}\label{t:simu_gamma}
\resizebox{\textwidth}{!}{
\begin{tabular}{lrrrrrrrrrrrr}
\toprule

\multicolumn{13}{l}{\bf Relative Bias}\\
$p_0$ & 0.001 & 0.01 & 0.05 & 0.1 & 0.2 & 0.3 & 0.4 & 0.5 & 0.6 & 0.7 & 0.8 & 0.9 \\
\midrule
$n$             & \multicolumn{12}{c}{\bf ME} \\
$10^4$ &-0.020 &0.001 &0.001	 &0.001	 &0.000	 &0.000	 &0.000	  &0.000	 &0.000	 &0.000	 &0.000	 &0.000\\
$10^5$ &0.000  &0.000 &0.000	 &0.000	 &0.000	 &0.000	 &0.000	  &0.000	 &0.000	 &0.000	 &0.000	 &0.000\\
$10^6$ &-0.002 &0.000 &0.000	 &0.000	 &0.000	 &0.000	 &0.000	  &0.000	 &0.000	 &0.000	 &0.000	 &0.000\\
\midrule
$n$            & \multicolumn{12}{c}{\bf KP} \\
$10^4$ &	-0.886	&0.313	&0.206	&0.085	&-0.003	&-0.027	&-0.029	&-0.021	&-0.011	&-0.001	&0.004	&0.004\\
$10^5$ &	-0.905	&0.313	&0.205	&0.084	&-0.003	&-0.027	&-0.029	&-0.022	&-0.011	&-0.001	&0.004	&0.004\\
$10^6$ &	-0.906	&0.313	&0.205	&0.084	&-0.003	&-0.027	&-0.029	&-0.022	&-0.011	&-0.001	&0.004	&0.004\\
\midrule
$n$          & \multicolumn{12}{c}{\bf VA} \\
$10^4$ &	11.773	&1.777	&0.236	&-0.012	&-0.107	&-0.096	&-0.062	&-0.029	&-0.008	&0.002	&0.003	&0.001\\
$10^5$ &	11.763	&1.775	&0.235	&-0.013	&-0.108	&-0.096	&-0.062	&-0.029	&-0.008	&0.002	&0.003	&0.001\\
$10^6$ &	11.761	&1.774	&0.235	&-0.013	&-0.108	&-0.096	&-0.062	&-0.029	&-0.008	&0.002	&0.003	&0.001\\
\midrule
$n$          & \multicolumn{12}{c}{\bf HGBRC} \\
$10^4$ &	-0.002	&-0.036	&-0.022	&-0.012	&-0.004	&-0.001	&0.001	&0.001	&0.001	&0.001	&0.000	&0.000\\
$10^5$ &	0.013	&-0.010	&-0.008	&-0.005	&-0.002	&-0.001	&0.000	&0.000	&0.000	&0.000	&0.000	&0.000\\
$10^6$ &	0.009	&-0.007	&-0.006	&-0.004	&-0.001	&0.000	&0.000	&0.000	&0.000	&0.000	&0.000	&0.000\\

\toprule
\multicolumn{13}{l}{\bf Relative RMSE}\\
$p_0$ & 0.001 & 0.01 & 0.05 & 0.1 & 0.2 & 0.3 & 0.4 & 0.5 & 0.6 & 0.7 & 0.8 & 0.9 \\
\midrule
$n$             & \multicolumn{12}{c}{\bf ME} \\
$10^4$ &	0.316	&0.097	&0.038	&0.024	&0.013	&0.009	&0.006	&0.004	&0.002	&0.001	&0.001	&0.000\\
$10^5$ &	0.103	&0.031	&0.012	&0.008	&0.004	&0.003	&0.003	&0.001	&0.001	&0.000	&0.000	&0.000\\
$10^6$ &	0.032	&0.009	&0.004	&0.002	&0.001	&0.001	&0.001	&0.000	&0.000	&0.000	&0.000	&0.000\\
\midrule
$n$             & \multicolumn{12}{c}{\bf KP} \\
$10^4$ &	0.891	&0.330	&0.209	&0.087	&0.011	&0.028	&0.029	&0.022	&0.011	&0.002	&0.004	&0.004\\
$10^5$ &	0.906	&0.315	&0.206	&0.084	&0.005	&0.027	&0.029	&0.022	&0.011	&0.001	&0.004	&0.004\\
$10^6$ &	0.907	&0.313	&0.205	&0.084	&0.003	&0.027	&0.029	&0.022	&0.011	&0.001	&0.004	&0.004\\
\midrule
$n$        & \multicolumn{12}{c}{\bf VA} \\
$10^4$ &	11.776	&1.777	&0.237	&0.019	&0.108	&0.096	&0.062	&0.030	&0.008	&0.002	&0.003	&0.001\\
$10^5$ &	11.763	&1.775	&0.235	&0.014	&0.108	&0.096	&0.062	&0.030	&0.008	&0.002	&0.003	&0.001\\
$10^6$ &	11.761	&1.774	&0.235	&0.013	&0.108	&0.096	&0.062	&0.030	&0.008	&0.002	&0.003	&0.001\\
\midrule
$n$          & \multicolumn{12}{c}{\bf HGBRC} \\
$10^4$ &	0.285	&0.113	&0.046	&0.027	&0.014	&0.008	&0.006	&0.004	&0.002	&0.001	&0.001	&0.000\\
$10^5$ &	0.090	&0.035	&0.015	&0.009	&0.004	&0.003	&0.002	&0.001	&0.001	&0.000	&0.000	&0.000\\
$10^6$ &	0.027	&0.013	&0.007	&0.004	&0.002	&0.001	&0.001	&0.000	&0.000	&0.000	&0.000	&0.000\\

\bottomrule
\end{tabular}
}
\caption*{\footnotesize See Table \ref{t:simu} for details.}
\end{table}

Our findings here are similar to those presented in Section \ref{sec:Simulation}.  
First, our proposed ME method performs very well in terms of both the bias and RMSE. 
They are smaller than those of the other methods for most of the $(n,p_0)$ combinations, especially the bias. 
Second, the KP and the VA methods both impose some parametric assumptions on the Lorenz curve and hence implicitly on the underlying density.  
The misspecification biases could be quite large depending on the true data generating process, especially for the VA method.

\end{document}